\documentclass[reqno,11pt,letterpaper]{amsart}
\usepackage{amsmath,amssymb,amsthm,graphicx,mathrsfs,url}
\usepackage[usenames,dvipsnames]{xcolor}
\usepackage[colorlinks=true,linkcolor=Red,citecolor=Green]{hyperref}
\usepackage{amsxtra}
\usepackage{wasysym} 
\usepackage{graphicx}
\usepackage{subcaption}
\usepackage{tikz}
\usepackage[normalem]{ulem}
\def\arXiv#1{\href{http://arxiv.org/abs/#1}{arXiv:#1}}

\def\?[#1]{\textbf{[#1]}\marginpar{\Large{\textbf{??}}}}

\def\smallsection#1{\smallskip\noindent\textbf{#1}.}
\let\epsilon=\varepsilon 

\newcommand{\RR}{{\mathbb R}}

\newcommand{\CC}{{\mathbb C}}
\newcommand{\TT}{{\mathbb T}}
\newcommand{\ZZ}{{\mathbb Z}}

\newtheorem{theo}{Theorem}
\newtheorem{prop}{Proposition}[section]	
\newtheorem{defi}[prop]{Definition}

\newtheorem{assumption}{Assumption}
\newtheorem{lemm}[prop]{Lemma}

\newtheorem{rem}{Remark}
\numberwithin{equation}{section}

\DeclareMathOperator{\Spec}{Spec}

\let\Im=\Imag

\DeclareMathOperator{\sgn}{sgn}

\DeclareMathOperator{\tr}{tr}

\def\indic{\operatorname{1\hskip-2.75pt\relax l}}
\usepackage{scalerel}

\newcommand\reallywidehat[1]{\arraycolsep=0pt\relax%
\begin{array}{c}
\stretchto{
  \scaleto{
    \scalerel*[\widthof{\ensuremath{#1}}]{\kern-.5pt\bigwedge\kern-.5pt}
    {\rule[-\textheight/2]{1ex}{\textheight}} 
  }{\textheight} %
}{0.5ex}\\           
#1\\                 
\rule{-1ex}{0ex}
\end{array}
}

\author{Simon Becker}
\address[Simon Becker]{ETH Zurich, 
Institute for Mathematical Research, 
Rämistrasse 101, 8092 Zurich, 
Switzerland}
\email{simon.becker@math.ethz.ch}

\author{Lingrui Ge}
\address[Lingrui Ge]{Beijing International Center for Mathematical Research, Peking University, Beijing, China}
\email{gelingrui@bicmr.pku.edu.cn}

\author{Jens Wittsten}
\address[Jens Wittsten]{Department of Engineering, University of Bor{\aa}s, SE-501 90 Bor{\aa}s, Sweden}
\title[Hofstadter butterflies and transport properties]{Hofstadter butterflies and metal/insulator transitions for moir\'e heterostructures}
\email{jens.wittsten@hb.se}

\begin{document}

\begin{abstract}
We consider a tight-binding model recently introduced by Timmel and Mele \cite{TM2020} for strained moir\'e heterostructures. We consider two honeycomb lattices to which layer antisymmetric shear strain is applied to periodically modulate the tunneling between the lattices in one distinguished direction. This effectively reduces the model to one spatial dimension and makes it amenable to the theory of matrix-valued quasi-periodic operators.  We then study the charge transport and spectral properties of this system, explaining the appearance of a Hofstadter-type butterfly and the occurrence of metal/insulator transitions that have recently been experimentally verified for non-interacting moir\'e systems \cite{W20}. For sufficiently incommensurable moir\'e lengths, described by a diophantine condition, as well as strong coupling between the lattices, which can be tuned by applying physical pressure, this leads to the occurrence of localization phenomena.
\end{abstract}

\maketitle  
\begin{figure}[h!] 
\includegraphics[width=5.2cm,height=5.2cm]{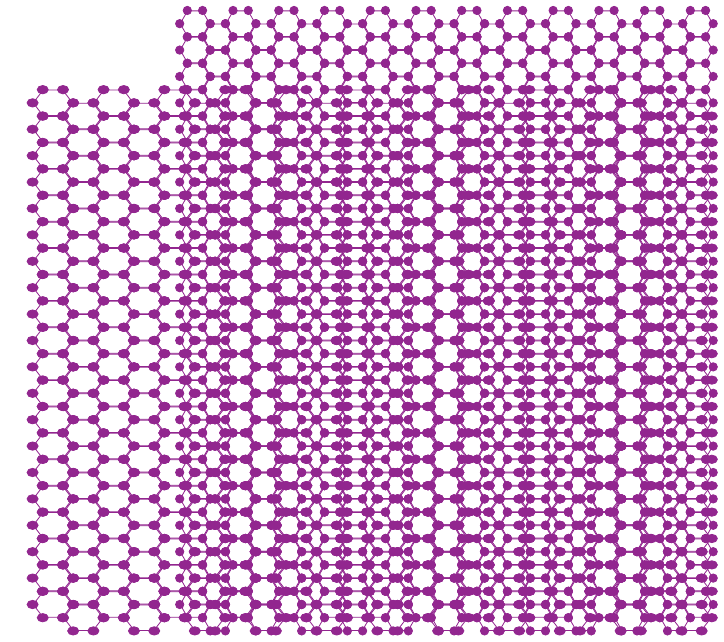} \qquad \includegraphics[width=5.2cm,height=5.2cm]{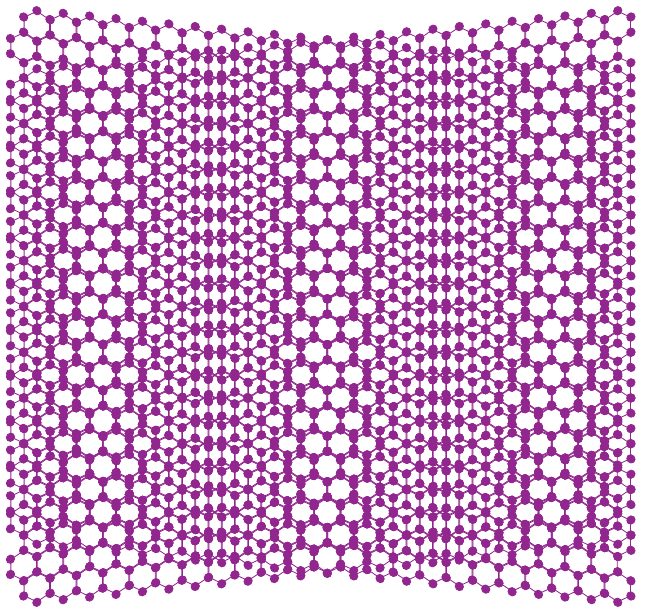}
\caption{Superposition of two honeycomb lattices. On the left, one of the lattices has been exposed to uniaxial strain in the horizontal direction. On the right, the lattices have been exposed to anti-symmetric shear strain. Both cases exhibit an effectively one-dimensional moir\'e pattern in the horizontal direction.\label{fig:two}}
\label{fig:2strains}
\end{figure}

\section{Introduction}

During the past few decades, scientists have learned how to prepare two-dimensional materials in the form of crystals that are only one atom or one molecular layer thick. When two-dimensional crystals are stacked on top of each other, with each layer offset or rotated with respect to the others, they form a large-scale interference pattern known as a {\it moir\'e pattern}.  If the starting crystals are semiconductors or semimetals, the low-energy electronic states are described by a Hamiltonian with the periodicity of the much larger moir\'e pattern instead of that of the original crystal, which influences the electronic properties of the material.

The most prominent example of a moiré material is twisted bilayer graphene (TBG) formed by stacking two sheets of graphene (a semimetal composed of carbon atoms arranged in a honeycomb lattice) on top of each other, offset by a small twist angle. At certain twist angles, called the {\it magic angles}, TBG exhibits an unconventional form of superconductivity, and the lowest band in the electronic band structure becomes flat. Mathematically, this corresponds to the Hamiltonian having a flat Bloch-Floquet band at zero energy \cite{BEWZ21,BEWZ22}.

If two or several lattice structures are stacked with incommensurable twisting angles between layers, the Hamiltonian describing the low-energy electronic states becomes quasi-periodic instead. The recently emerging field of twistronics has provided a variety of examples of such quasi-periodic Hamiltonians, but to our knowledge there are no mathematical results on any of them so far. While such examples exhibit quasi-periodicity in two spatial dimensions, we shall restrict us here to lattice structures that are quasi-periodic in at most one direction and periodic in the orthogonal direction. Such lattice systems appear naturally by superimposing strained two-dimensional lattices as seen in Figure \ref{fig:2strains}, and they are important because in contrast to the twist angle -- which is set during nanofabrication -- strain can be tuned in situ. 

Specifically, we shall study a one-dimensional armchair model for bilayer graphene proposed in \cite{TM2020}. Due to periodic strain-modulation, the bilayer graphene is periodic in one direction and, depending on the arithmetic properties of the strain, either periodic or quasi-periodic in the orthogonal direction exhibiting a moir\'e pattern, visible along the horizontal direction in Figure \ref{fig:2strains}. We will refer to this as the moir\'e direction. Using the periodicity in the direction orthogonal to the moir\'e direction, Floquet theory then provides a family of one-dimensional Hamiltonians depending on some quasi-momentum $\theta \in \RR/ \ZZ$. By analyzing this family of Hamiltonians we are able to study strained bilayer graphene as it transitions between having metal-like conductivity and insulator-like conductivity (so-called metal-insulator transitions) which is characterized by the Hamiltonian having no point spectrum or pure point spectrum, respectively.

Studying such an effectively one-dimensional setting has the advantage that the well-developed theory of quasi-periodic one-dimensional discrete operators is applicable. 
Following the groundbreaking paper \cite{ds} the theory of quasi-periodic Schr\"odinger
operators has been developed extensively in the past 40 year, see
\cite{B,damanik,jcong,jm,you} for surveys of more recent results. For the one-dimensional case, a key breakthrough is  Avila's global theory \cite{avila}, along with its quantitative version \cite{gjyz}, allowing powerful applications \cite{g,gjy,gj}.

Similar to magnetic fields, moir\'e structures have been shown to exhibit fractal spectra, the so-called Hofstadter butterfly \cite{BM11,CNKM20,L21} (compare with Figures \ref{fig:AC} and \ref{fig:fat_cantor} below), and metal-insulator transitions \cite{C18,W20}.
Metal-insulator transitions have been observed by changing the tunneling rate when compressing the lattices or adapting the parameters of optical analogues. Experimental and theoretical studies of charge transport properties for one-dimensional moir\'e structures have been considered in \cite{BLB16}. We also discuss model operators for analogous results in two spatial dimensions. In contrast to some of the two-dimensional twisted moir\'e superstructures, one-dimensional models do not exhibit flat bands (demonstrated e.g.~by our Proposition \ref{prop:DOS_cont}). However, in \cite{BW22} it has been shown using semiclassical analysis techniques that the operator exhibits \emph{almost flat bands} characterized in terms of quasimodes.

\smallsection{Notation}
The identity and zero matrix of size $n$ are denoted by $I_n$ and $0_n$ respectively. (If it is clear from context we will sometimes simply write $0$ for the zero matrix of size $n$.)
Let $\mathcal H$ be a Hilbert space, then we also write $I_{\mathcal H}$ and $0_{\mathcal H}$ to denote the identity and zero operator on that space, respectively. 
We denote the four basic Pauli matrices by 
\[ \sigma_0 = I_{2}, \ \sigma_1 = \begin{pmatrix}0 & 1 \\ 1 & 0 \end{pmatrix}, \ \sigma_2 = \begin{pmatrix} 0 & - i \\ i & 0 \end{pmatrix}, \ \sigma_3 = \operatorname{diag}(1,-1).\]
We write $\operatorname{diag}(a_1,\ldots,a_n)$ to mean the diagonal $n\times n$ matrix with entries $a_1,\ldots,a_n$ on the diagonal.
We also write $[N]:=\{1,\ldots,N\}.$  We let $\mathbb R_{0}^+$ denote the positive real numbers including $0$, i.e.~$\mathbb R_0^+=[0,\infty)$, while we write $\mathbb R^+=(0,\infty)$. We also define $\langle n \rangle:=(1+\vert n \vert^2)^{1/2}.$ Finally, we let $\mathbb T^d :=\mathbb R^d/\mathbb Z^d$ be the $d$-dimensional torus and just write $\mathbb T$ for $\mathbb T^1.$

\subsection{Main results and organization of the article}

Following \cite{TM2020}, the dynamics of bilayer graphene under anti-symmetric shear strain is approximated using a tight-binding (that is, discrete) kinetic Dirac operator $D_{\operatorname{kin}}(\theta):\ell^2(\mathbb{Z};\mathbb{C}^4)\rightarrow\ell^2(\mathbb{Z};\mathbb{C}^4)$ which is defined in terms of $\gamma_{15} = \operatorname{diag}(\sigma_1,\sigma_1)$, $\gamma_{25} = \operatorname{diag}(\sigma_2,\sigma_2)$,
with Pauli matrices $\sigma_i$, as 
\[(D_{\operatorname{kin}}(\theta)\psi)_n = t(\theta) \psi_{n+1} + t(\theta) \psi_{n-1} + t_0 \psi_n, \]
where $t(\theta) =\cos(2\pi \theta) \gamma_{15}+\sin(2\pi \theta) \gamma_{25}$, with $\operatorname{det}(t(\theta)) = 1,$ $\Vert t(\theta) \Vert =1,$ and $t_0 = \gamma_{15}$. Here $\theta$ indicates the quasi-momentum perpendicular to the moir\'e direction.

The two honeycomb lattices interact through a tunneling interaction. To define it, introduce scalar functions
\begin{equation}\label{eq:scalarpotentials}
\begin{aligned}
    U(x)&:=\frac{1+2 \cos(2\pi x)}{3}\\
    U^{\pm}_{\operatorname{c}}(x)&:=\frac{1- \cos(2\pi x)\pm\sqrt 3 \sin(2\pi x)}{3} 
    =\frac{1-2\cos(2\pi x\pm\pi/3)}{3}
\end{aligned}    
\end{equation}
and matrix-valued tunneling potentials $V_c, V_{\operatorname{ac}}^\phi$, defined by
\begin{equation*}
\begin{split}
V_{\operatorname{ac}}^\phi&= \begin{pmatrix} 0_2 &W_{\operatorname{ac}}^\phi\\W_{\operatorname{ac}}^\phi & 0_2 \end{pmatrix} \text{ with } W_{\operatorname{ac}}^\phi(x) = \operatorname{diag}(U(x-\tfrac{\phi}{L}),U(x+\tfrac{\phi}{L})), \\
V_{\operatorname{c}}&=\begin{pmatrix} 0_2 & W_{\operatorname{c}} \\ W_{\operatorname{c}}^* & 0_2 \end{pmatrix} \text{ with }
W_{\operatorname{c}} = \begin{pmatrix} 0 &U^{-}_{\operatorname{c}}\\ U^{+}_{\operatorname{c}}& 0 \end{pmatrix}.
\end{split}
\end{equation*}
For coupling strengths $w = (w_0,w_1) \in \RR_0^+ \times \RR_0^+$ the tunneling interaction is then given by
\begin{equation}\label{eq:Vw}
    V_w^\phi(x)=w_0 V_{\operatorname{ac}}^\phi(x) + w_1 V_{\operatorname{c}}(x).
\end{equation}
Inspired by the Bistritzer-MacDonald model \cite{BM11} we refer to the first summand as the \emph{anti-chiral} part, which describes tunneling between A-A$'$/B-B$'$ atoms. The second summand is the  
\emph{chiral} part modelling the tunneling between A-B$'$/B-A$'$ atoms. Here, A and B correspond to the two different representatives of the fundamental cell of a honeycomb lattice with the prime $'$ indicating atoms of the second lattice. The parameter $\phi$ incorporates the different tunneling amplitude for A-A$'$ and B-B$'$ sites due to their dislocation in space.
The Hamiltonian $H_{w}(\theta,\phi,\vartheta):\ell^2(\ZZ;\CC^4) \rightarrow \ell^2(\ZZ;\CC^4)$ is then, for some fixed length $L>0$ of the moir\'e cell, given by the Dirac-Harper model
\[ (H_{w}(\theta,\phi,\vartheta)\psi)_n  = (D_{\operatorname{kin}}(\theta)\psi)_n + V^\phi_w(\vartheta + \tfrac{n}{L})\psi_n,\]
where $\vartheta \in [0,1/L]$.  
Note however that since the Hamiltonian is invariant under integer translations of $1/L$ we may assume without loss of generality that $1/L \in \mathbb T$, and thus take $\vartheta \in [0,1].$ The length of the fundamental cell is related to the strength of the strain. In \cite{TM2020} it is assumed that $1/L$ is rational so that the Hamiltonian becomes periodic, but here we shall allow any $1/L\in \mathbb T$ and study how arithmetic properties of the length affects the properties of the system. 
Unlike for the almost Mathieu operator, the only physically relevant frequency in the tunneling potential, as introduced by \cite{TM2020}, is $\vartheta =0.$ However, we introduce the parameter $\vartheta$ for our mathematical analysis. Physically it corresponds to an additional offset between the lattices at the origin. Similarly, the original model \cite{TM2020} only considers $\phi=0$, but allowing for different values of $\phi$ turns out to be convenient for the mathematical presentation. The only part of the paper in which non-zero $\phi$ plays an important role is in \S\ref{sec:AVAL}, see ~Assumption \ref{ass:assumption} in that subsection. 
\begin{rem}
We shall occasionally suppress the parameter dependence in the Hamiltonian and related quantities to simplify the notation. In particular, we sometimes write $V_w$ instead of $V_w^\phi$ for the potential in \eqref{eq:Vw}, and will sometimes for convenience use the notation
\begin{equation}\label{eq:acplusminus}
    U^{\pm}_{\operatorname{ac}}(x):=\frac{1+2 \cos(2\pi (x\pm\phi/L))}{3}  
\end{equation}
so that $W_{\operatorname{ac}}^\phi(x) = \operatorname{diag}(U^{-}_{\operatorname{ac}}(x),U^{+}_{\operatorname{ac}}(x))$.
\end{rem}

Introducing the shift operator $\tau \psi_n :=\psi_{n-1}$, the Hamiltonian takes the form of a $4$ by $4$ matrix-valued discrete operator that reads in terms of $K(\theta):=1+ e^{-2\pi i \theta} (\tau + \tau^*)$
\begin{equation}
 \label{eq:Hamiltonian}
 H_{w}(\theta,\phi,\vartheta) = \begin{pmatrix} 0 & K(\theta)& w_0 U(\vartheta+\tfrac{\bullet-\phi}{L}) & w_1 U_{\operatorname{c}}^-(\vartheta+\tfrac{\bullet}{L}) \\ 
  K(\theta)^* &0 &w_1 U_{\operatorname{c}}^+(\vartheta+\tfrac{\bullet}{L}) & w_0 U(\vartheta+\tfrac{\bullet+\phi}{L})   \\ 
 w_0 U(\vartheta+\tfrac{\bullet-\phi}{L})  &w_1 U_{\operatorname{c}}^+(\vartheta+\tfrac{\bullet}{L}) &0 & K(\theta)\\ 
w_1 U_{\operatorname{c}}^-(\vartheta+\tfrac{\bullet}{L}) &w_0 U(\vartheta+\tfrac{\bullet+\phi}{L})   &   K(\theta)^*&0   \end{pmatrix}. 
\end{equation}
   When $w_0 \equiv 0$ we call this the \emph{chiral model} and when $w_1 \equiv 0$ the \emph{anti-chiral model}.
We shall mostly focus on the full model with both coupling parameters $w_0,w_1$ being nonzero.

After introducing the framework of matrix-valued cocycles, which can be found for example in \cite{AJS15,jm}, we discuss the case of moir\'e lengths $L$ that are rational or close to rational numbers in which case point spectrum is absent, see Propositions \ref{prop:rational} and \ref{prop:Liouville}. This extends -- with obvious modifications -- a classical theorem by Gordon \cite{G76} (discrete operators) and Simon \cite{S82} (continuous operators) to matrix-valued operators.
 
In the opposite regime of diophantine moir\'e length scales, we prove that if the coupling between the two lattice structures is strong enough, then the Hamiltonian exhibits so-called {\it Anderson localization}, i.e., pure point spectrum with exponentially decaying eigenfunctions, see Theorem \ref{theo:Aloc1}. This can be experimentally seen by the application of physical pressure. In contrast, if the coupling between the lattices is sufficiently weak, transport and absolutely continuous spectrum that is present in case of non-interacting graphene sheets persist, see Theorem \ref{theo:AC}. This localization argument relies on a matrix generalization of the theory that has been obtained by Klein \cite{K17} extending earlier works \cite{BGS02}.

\begin{figure}
\includegraphics[width=6cm,height=4.7cm]{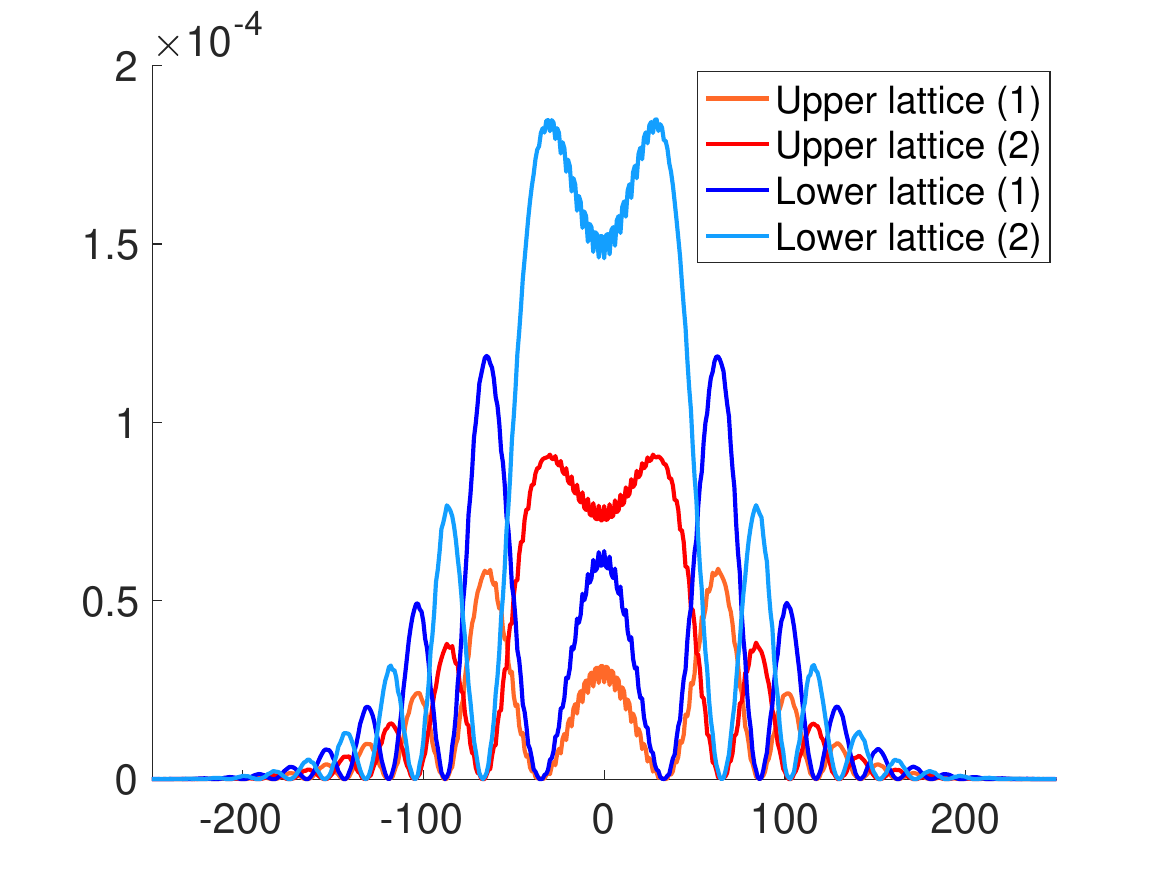}\includegraphics[width=6cm,height=4.7cm]{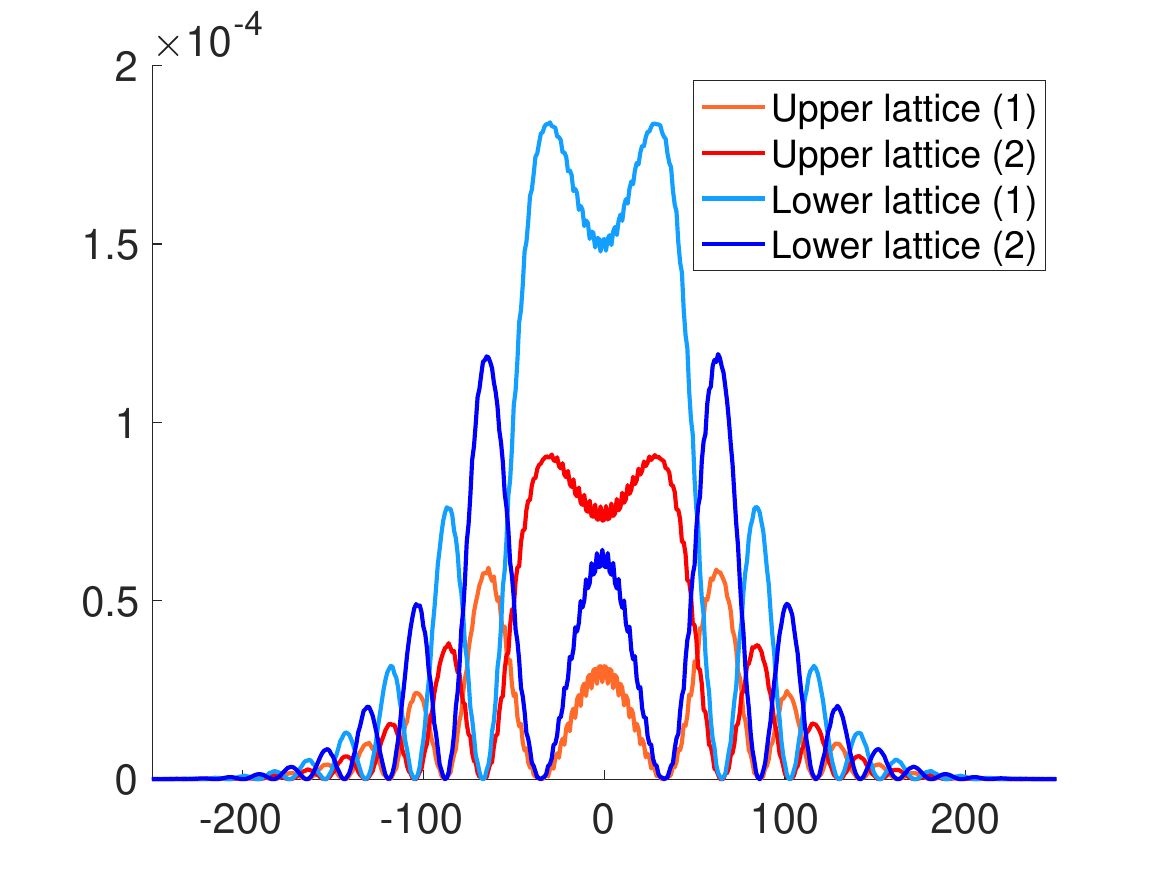}
\\ 
\includegraphics[width=6cm,height=4.7cm]{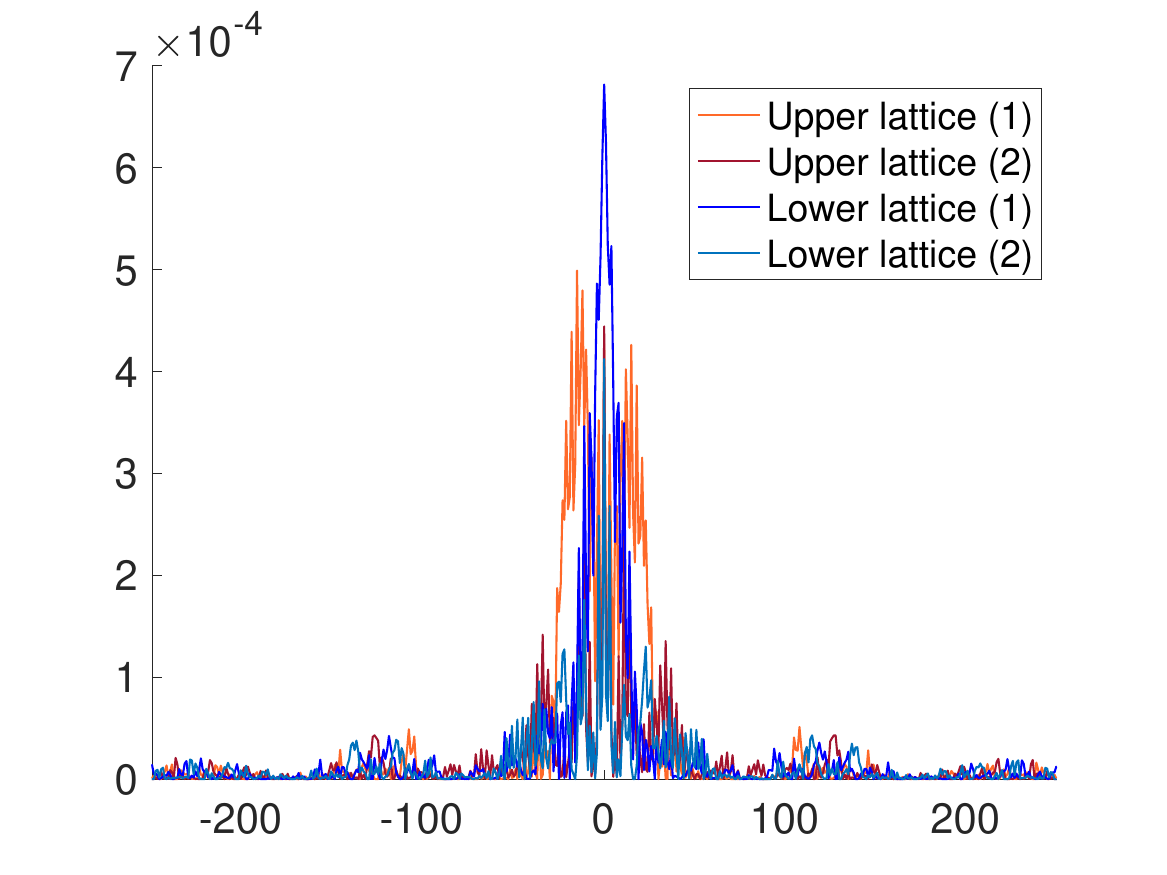}\includegraphics[width=6cm,height=4.7cm]{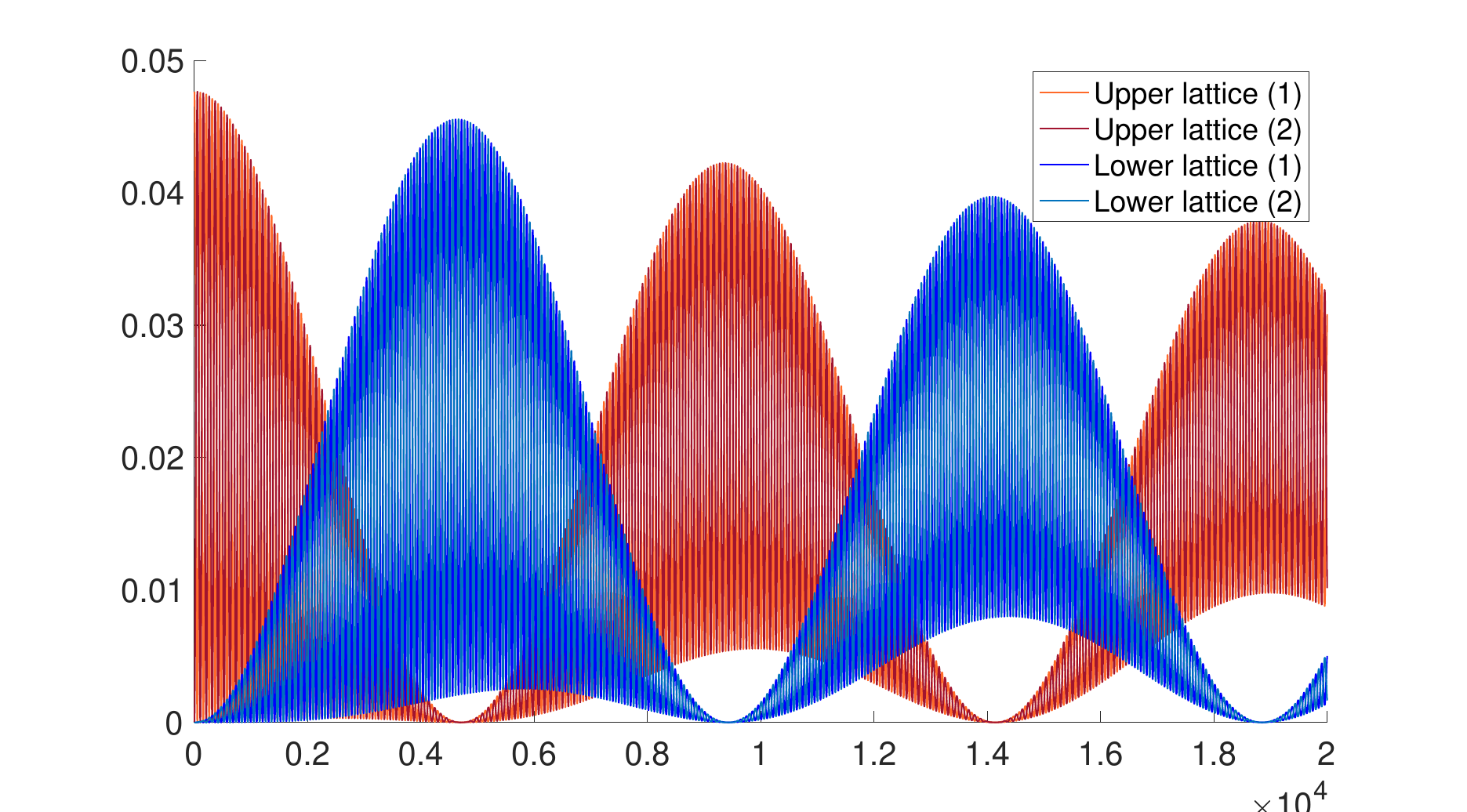}
\caption{Time-evolution of Gaussian wavepackets. (Top row) Time-evolved discretized Gaussian state for chiral (left) and anti-chiral (right) model with $L=3$ with weak coupling $w_0=0.1$ after time $T_{\text{fin}}=5000$ (Space-Amplitude plot). Gaussian state for strong coupling $w_1=1.9$, $w_0=0$, and $L=\pi$ on the lower left figure for the chiral model after $T_{\text{fin}}=2 \cdot 10^4$. Localization effects are clearly visible. On the lower right we see the time evolution (Time-Amplitude plot) corresponding to the amplitude for a Gaussian wavepacket started at the upper lattice. Here, $(1)$ and $(2)$ refer to the respective components labeling atoms of type A and B, respectively. The wavepacket oscillates between the different layers.}
\label{fig:dynamics}
\end{figure}

To illustrate the absolutely continuous and localized regimes of the Hamiltonian for diophantine $L$, Figure \ref{fig:dynamics} describes the time-evolution of a discrete Gaussian-type wave-packet $$\psi_n = \frac{e^{-n^2/(2\sigma^2)}}{\sqrt{2\pi \sigma^2}}e_1$$ under the Schr\"odinger dynamics $\psi(t)=e^{-iHt}\psi,$ where $e_1$ is the first unit vector and $\sigma = \sqrt{70}.$ Aside from the lower left panel, where we used a strong coupling with incommensurable length scale, the coupling used for the production of all remaining panels is weak. The top two panels clearly show a spread of the Gaussian wavepacket (absolutely continuous regime with metal-like transport resulting from weak coupling), while the bottom left panel shows that the wavepacket remains fairly localized under the time evolution (localized regime with insulator-like transport resulting from strong coupling). In addition, the bottom right panel shows that the layer in which the wavepacket is concentrated alternates between top and bottom under time evolution.

To establish Anderson localization we rely on having a positive lower bound on the Lyapunov exponents for a matrix-valued cocycle associated to $H_w$. The lower bound is obtained in Proposition \ref{prop:LEsfull} using a complexification argument, and for this to work we have to assume that $w_0 \neq w_1$. When $w_0=w_1$ there are similar obstructions to using the results of \cite{K17} (or \cite{DK14}) to obtain positive lower bounds on the Lyapunov exponents, see Remark \ref{rem:constantEV}. We note that although the model proposed by Timmel and Mele \cite{TM2020} has $w_0=w_1$, it is not very likely for physically realistic parameters that $w_0$ is exactly equal to $w_1$ (for twisted bilayer graphene it has for example been experimentally verified that $w_0/w_1\approx$ $0.7$-$0.8$ for small twisting angles \cite[p.~5]{magic} due to lattice relaxation effects). In this sense the assumption $w_0\ne w_1$ is not very restrictive. Nevertheless, studying localization when $w_0=w_1$ remains an interesting open mathematical question.

Unlike in the case of magnetic fields, twisted lattice systems do in general not allow for an explicit reduction to one-dimensional quasi-periodic operators. Since the theory of multi-dimensional quasi-periodic operators is far less developed, this limits the tools available to understand fractal spectra, see Proposition \ref{prop:AMO}, and metal/insulator transitions in depth. Most results in higher dimensions are limited to establishing Anderson localization for, in our case, sufficiently strong coupling \cite{BGS02}. 
  
Methods and results showing Cantor spectrum are largely limited to scalar one-dimensional quasi-periodic operators \cite{aj1,GS11} and also in our case, we rely on the diagonalizability of the matrix-valued operator in one of the two considered cases to establish fractal spectrum.

Our article is structured as follows:

\smallsection{Outline of article}
\begin{itemize}
    \item In Section \ref{sec:Basic_prop} we discuss basic spectral properties including Lyapunov exponents of the Dirac-Harper model \eqref{eq:Hamiltonian}.
\item In Section \ref{sec:ratio} we study with Propositions \ref{prop:rational} and \ref{prop:Liouville} the spectral and transport properties for moir\'e lengths $1/L$ that are rational or only mildly irrational numbers, i.e. \emph{Liouville numbers}. We also establish the absence of flat bands for commensurable length scales, see Proposition \ref{prop:DOS_cont}. 
\item In Section \ref{sec:AL} we study the regime of strongly irrational (diophantine) moir\'e lengths $1/L$ that satisfy a diophantine condition. For strong tunneling interaction, this is the regime of Anderson localization and insulator-like conductivity (absence of transport) proven in Theorem \ref{theo:Aloc1}. For comparison, we discuss arithmetic Anderson localization in \S\ref{sec:AVAL} which requires us to modify the Hamiltonian slightly (as described by Assumption \ref{ass:assumption}), see Theorem \ref{theo:Jitomirskaya}. 

\item In Section \ref{sec:AC} we show the existence of absolutely continuous spectrum for weak coupling of the lattices, see Theorem \ref{theo:AC}. This is the regime of metal-like conductivity.
\item In Section \ref{sec:Cantor} we show the existence of Cantor spectrum for the Dirac-Harper operator \eqref{eq:Hamiltonian} in the anti-chiral limiting case, see Proposition \ref{prop:AMO}. 
\item In Section \ref{sec:spec_ana} we study the spectrum and prove absence of flat bands in an effective low-energy model.
\item In Section \ref{sec:2D} we give an outlook on 2D generalizations of the Dirac-Harper model \eqref{eq:Hamiltonian} for twisted square lattices.
\end{itemize}

\section{Basic properties of the Dirac-Harper model} 
In this section we start by exhibiting some basic spectral properties of the Dirac-Harper Hamiltonian \eqref{eq:Hamiltonian}.
\label{sec:Basic_prop}
 \begin{lemm}\label{lem:particlehole}
 In case of the limiting chiral ($w_0=0$) and anti-chiral ($w_1=0$) model, the Hamiltonian satisfies particle-hole symmetry, i.e.
 \[ \Spec(H_w(\theta,\phi,\vartheta)) = - \Spec(H_w(\theta,\phi,\vartheta)).\]
 \end{lemm}
 \begin{proof}
For both the chiral and anti-chiral limiting model, this follows by conjugating the Hamiltonian by unitaries
 \[ P_{\operatorname{c}}:=\begin{pmatrix} 1& 0 & 0 & 0 \\ 0 & 0 & 1 & 0 \\ 0 & -1 & 0 & 0 \\  0 & 0 & 0 &-1 \end{pmatrix}\text{ and }P_{\operatorname{ac}}:=\begin{pmatrix} ie^{i\frac{\pi}{4}}& 0 & 0 & 0 \\ 0 & 0 & 0 & -e^{-i\frac{\pi}{4}} \\ 0 &0 & -ie^{i\frac{\pi}{4}} & 0  \\  0 &  e^{-i\frac{\pi}{4}}&0 &0 \end{pmatrix}\] respectively from the left which turns the Hamiltonian into a block off-diagonal operator. Then conjugating with $\operatorname{diag}(I_{2},-I_{2})$ implies the claim.
 \end{proof}
 \subsection{Ergodic properties of the system}
The spectral and dynamical properties of the system are governed by the arithmetic properties of $L$, foremost depending on whether $L \in \mathbb Q^{+}$ (periodic) or $L \in \mathbb R^{+} \setminus \mathbb Q$ (quasi-periodic).
Let $\psi$ be a solution to $H_{w}(\theta,\phi,\vartheta)\psi = \mathscr E \psi$, with $\mathscr E \in \CC$, then we can write the solution as
\begin{equation}\label{eq:solution}
    \begin{pmatrix} \psi_{n+1} \\ \psi_n \end{pmatrix} = \begin{pmatrix} t(\theta)^{-1}(\mathscr E- t_0 - V_w^\phi(\vartheta+ \tfrac{n}{L})) & - I_4  \\ I_4 & 0 \end{pmatrix} \begin{pmatrix} \psi_{n} \\ \psi_{n-1} \end{pmatrix}.
\end{equation} 
Since $t(\theta)$ is self-involutive, $t(\theta)=t(\theta)^{-1}$, we find that the associated Schr\"odinger cocycle $(1/L, A^{\mathscr E,\theta,w})$ where $L>0$ and $A^{\mathscr E,\theta,w} \in C^{\omega}(\RR/\ZZ,  \operatorname{SL}(8,\CC))$ is given as
\begin{equation}
\label{eq:full_cocycle}
A^{\mathscr E,\theta,w}(\vartheta) := \begin{pmatrix} Q^{\mathscr E,\theta,w}(\vartheta)& - I_4 \\ I_4  & 0 \end{pmatrix},
 \end{equation}
where $Q^{\mathscr E,\theta,w}(\vartheta):=t(\theta)(\mathscr E- t_0 -V_w^\phi(\vartheta)).$
Note that $A^{\mathscr E,\theta,w}(\vartheta)$ is invertible since \eqref{eq:full_cocycle} and the determinant formula for block matrices implies that $\det(A^{\mathscr E,\theta,w}(\vartheta))=1$.
For $w_1=0$, we denote the cocycle by $A^{\mathscr E,\theta,w_0}_{\operatorname{ac}}$ and for $w_0=0$ by $A^{\mathscr E,\theta,w_1}_{\operatorname{c}}.$
The top left block matrix of the cocycle \eqref{eq:full_cocycle} reads 
\begin{equation}
\label{eq:Q}
 Q^{\mathscr E,\theta,w}= \begin{pmatrix}-e^{-2\pi i \theta} & \mathscr E e^{-2\pi i \theta}  & -e^{-2\pi i \theta} w_1 U^{+}_{\operatorname{c}} & -e^{-2\pi i \theta}  w_0 U^{+}_{\operatorname{ac}} \\ \mathscr E e^{2\pi i \theta} & - e^{2\pi i \theta} & -e^{2\pi i \theta}  w_0U^{-}_{\operatorname{ac}}  &-e^{2\pi i \theta} w_1 U^{-}_{\operatorname{c}}   \\  -e^{-2\pi i \theta} w_1 U^{-}_{\operatorname{c}} & -e^{-2\pi i \theta}  w_0U^{+}_{\operatorname{ac}} & -e^{-2\pi i \theta} & \mathscr E e^{-2\pi i \theta}  \\   -e^{2\pi i \theta} w_0U^{-}_{\operatorname{ac}}  &  -e^{2\pi i \theta}  w_1 U^{+}_{\operatorname{c}} & \mathscr E e^{2\pi i \theta} & - e^{2\pi i \theta} \end{pmatrix},
 \end{equation}
 where $U^{\pm}_{\operatorname{ac}}$ is defined in accordance with \eqref{eq:acplusminus}.
Let $L >0,$ and introduce the shift $Tx := x+1/L$ and the $n$-th cocycle iterate
\begin{equation}\label{eq:cocycle}
    A_n^{\mathscr E,\theta,w}(\vartheta):= \prod_{i=n}^{1} A^{\mathscr E,\theta,w}(T^i \vartheta) = A^{\mathscr E,\theta,w}( \vartheta+\tfrac{n}{L}) \cdots A^{\mathscr E,\theta,w}( \vartheta+\tfrac{1}{L})
\end{equation}
so that \eqref{eq:solution} takes the form
\[\begin{pmatrix} \psi_{n+1} \\ \psi_n \end{pmatrix} = A_{n}^{\mathscr E,\theta,w}(\vartheta) \begin{pmatrix} \psi_1 \\ \psi_{0} \end{pmatrix}.\]

In this work, we are mostly concerned with irrational length scales $L \in \mathbb R^+ \backslash \mathbb Q.$ In addition, we may assume without loss of generality that $1/L \in \mathbb T,$ since the Hamiltonian is invariant under integer translations of $1/L.$
Under this assumption of irrational $L$, the dynamics becomes ergodic and Oseledet's theorem, see e.g. \cite{AJS15} and references therein, implies the existence of eight (possibly degenerate) Lyapunov exponents (LEs) $L_j=L_j(1/L,A^{\mathscr E,\theta,w}) \in [-\infty,\infty)$.
Arranging them as $L_1 \ge\ldots\ge L_8$ 
repeated according to multiplicity, they are given by 
\begin{equation}\label{eq:LEdef}
L_i(1/L,A^{\mathscr E,\theta,w}) = \lim_{n \to \infty} \frac{1}{n} \int_{\RR/\ZZ} \log(\sigma_i(A^{\mathscr E,\theta,w}_n(\vartheta))) \, d\vartheta,
\end{equation}
where $\sigma_k(B)$ is the $k$-th singular value of a matrix $B$, with the convention that $\sigma_k(B) \ge \sigma_{k+1}(B).$ 
We also define $L^k(1/L,A^{\mathscr E,\theta,w})= \sum_{j=1}^k  L_j(1/L,A^{\mathscr E,\theta,w})$.
The LEs $\mathbb R \ni t \mapsto L^k(1/L,A^{\mathscr E,\theta,w}(\bullet+it))$ are then convex and piecewise affine functions \cite{AJS15}. We emphasize that this property may however not be true for Lyapunov exponents $ L_k(1/L,A^{\mathscr E,\theta,w}).$

Regarding the Lyapunov exponents of the cocycle, we make the following simple observation using the symplectic structure of our cocycle.
\begin{lemm}\label{lem:LEpair}
For any $\mathscr{E}\in \mathbb{R}$, the LEs of $A^{\mathscr E,\theta,w}$ given by $ L_n(1/L,A^{\mathscr E,\theta,w})$ for $n \in \{1,\ldots,8\}$ appear in pairs satisfying  $$ L_{9-n}(1/L,A^{\mathscr E,\theta,w})=- L_n(1/L,A^{\mathscr E,\theta,w}),\quad n \in \{1,\ldots,4\}.$$
\end{lemm}
\begin{proof}
Since $t(\theta)^*=t(\theta)$ and $(t_0-V_{w}^\phi(\vartheta))^*=t_0-V_{w}^\phi(\vartheta)$, a simple computation shows that for any $\mathscr{E}\in \mathbb{R}$,
$$
A^{\mathscr E,\theta,w}(x)^* \Omega A^{\mathscr E,\theta,w}(x) = \Omega,
\quad \Omega = \begin{pmatrix} 0 & t(\theta) \\ -t(\theta) & 0 \end{pmatrix}.
$$
For the cocyle \eqref{eq:cocycle} this means upon iterating this identity that
\begin{align*}
A_n^{\mathscr E,\theta,w}(x)^* \Omega A_n^{\mathscr E,\theta,w}(x)& =
A_{n-1}^{\mathscr E,\theta,w}(x)^* \Big[A^{\mathscr E,\theta,w}(x+\tfrac{n}{L})^* \Omega A^{\mathscr E,\theta,w}(x+\tfrac{n}{L})\Big] A_{n-1}^{\mathscr E,\theta,w}(x)\\
&=A_{n-1}^{\mathscr E,\theta,w}(x)^*  \Omega A_{n-1}^{\mathscr E,\theta,w}(x)=\ldots=A^{\mathscr E,\theta,w}(x+\tfrac{1}{L})^*  \Omega A^{\mathscr E,\theta,w}(x+\tfrac{1}{L})\\
&=\Omega.    
\end{align*}
Thus, $(A_n^{\mathscr E,\theta,w}(x))^{-1}= \Omega^{-1}A_n^{\mathscr E,\theta,w}(x)^* \Omega.$ On the other hand, $\Omega$ is anti self-adjoint, so the argument also applies to the adjoint of $A_n^{\mathscr E,\theta,w}(x)$ and then also to the product $A_n^{\mathscr E,\theta,w}(x)^* A_n^{\mathscr E,\theta,w}(x)$. Therefore,  $\lambda$ is an eigenvalue of $A_n^{\mathscr E,\theta,w}(x)^* A_n^{\mathscr E,\theta,w}(x)$ if and only if $1/\lambda$ is. Since the eigenvalues of $A_n^{\mathscr E,\theta,w}(x)^* A_n^{\mathscr E,\theta,w}(x)$ are the squared singular values of $A_n^{\mathscr E,\theta,w}(x)$, the result follows in view of the ordering $ L_1\ge\ldots\ge L_8$.
\end{proof}
We also recall the characterization of the AC spectrum due to Kotani and Simon \cite{KS} showing that
\begin{equation}
\label{eq:Kotani-Simon}
S_j = \{ \mathscr E \in \mathbb R: \text{ there are }2j\le 8\text{ LEs } L_n(1/L,A^{\mathscr E,\theta,w}) \text{ that vanish}\},
\end{equation}
is the essential support of the absolutely continuous spectrum of multiplicity $2j.$ In particular, if $S_{4}$ contains an open interval $I$, then the spectrum of the Hamiltonian is purely absolutely continuous on $I.$

We now turn to estimates of LEs. To this end, note that
\[\lVert A^{\mathscr E,\theta,w}(\vartheta) x\rVert^2=\lVert Q^{\mathscr E,\theta,w}(\vartheta)x_1-x_2\rVert^2+\lVert x_1\rVert^2\quad \text{ if }x=(x_1,x_2)^\intercal\text{ with }x_j\in\CC^4.\]
     Since
     \[\lVert Q^{\mathscr E,\theta,w}(\vartheta)x_1-x_2\rVert^2\le \lVert Q^{\mathscr E,\theta,w}(\vartheta)x_1\rVert^2+\lVert x_2\rVert^2+2\lVert Q^{\mathscr E,\theta,w}(\vartheta)x_1\rVert \lVert x_2\rVert,\] 
     it is then easy to check that
\begin{align*}
\lVert A^{\mathscr E,\theta,w}(\vartheta)\rVert^2&\le  (1+\lVert Q^{\mathscr E,\theta,w}(\vartheta)\rVert)^2.
\end{align*}
Using the definition \eqref{eq:Vw} of $V_w^\phi$ it is easy to see that $\lVert V_w^\phi\rVert\le |w_0|\lVert V_\mathrm{ac}^\phi\rVert+|w_1|\lVert V_\mathrm{c}\rVert\le |w_0|+|w_1|$, and since $\lVert t(\theta)\rVert=\lVert t_0\rVert=1$, we get $\lVert Q^{\mathscr E,\theta,w}\rVert\le \lVert t(\theta)\rVert\lVert\mathscr E-\mathrm{t}_0-V_w^\phi\rVert\le |\mathscr E|+1+|w_0|+|w_1|$. This gives 
\begin{equation}\label{eq:normboundcocycle}
\lVert A^{\mathscr E,\theta,w}(\vartheta)\rVert\le 1+\lVert Q^{\mathscr E,\theta,w}(\vartheta)\rVert\le 2+|\mathscr E|+|w_0|+|w_1|    
\end{equation}
which implies upper bounds on the LEs
\begin{equation}
\label{eq:upper_bds}
  L_i(1/L,A^{\mathscr E,\theta,w})\le \log\left(2+\vert \mathscr E \vert +  \vert w_0\vert +\vert w_1\vert\right), \quad i=1,\ldots,4.
 \end{equation}

\subsection{Complexification of LEs}

After having stated upper bounds on LEs in \eqref{eq:upper_bds}, our first proposition gives lower bounds on LEs. By Kotani-Simon theory, strict positivity of the first four Lyapunov exponents, $ L_i (1/L, A^{\mathscr E,\theta,w}) >0$ for all $i \in \{1,\ldots,4\}$, implies the absence of absolutely continuous spectrum, cf.~the discussion around \eqref{eq:Kotani-Simon}. For the statement we let
\begin{equation}\label{eq:mu}
    \mu(w):=\sqrt{w_0^2-w_1^2\sin^2(2\pi\theta)}
\end{equation}
with branch chosen so that $\mu(0,w_1)=iw_1\sin(2\pi\theta)$, and introduce
\begin{equation}\label{eq:lambdafullmodel}
\begin{aligned}
  \lambda_{1}(w) &=  \frac13(w_1 \cos (2 \pi  \theta)+ \mu), 
   && \lambda_{2}(w) = \frac13(w_1 \cos (2 \pi  \theta)- \mu) , \\
  \lambda_{3}(w) &= \frac13(-w_1 \cos (2 \pi  \theta)+ \mu) ,
   && \lambda_{4}(w) = \frac13(-w_1 \cos (2 \pi  \theta)- \mu) .
\end{aligned}
\end{equation}
Observe that $\lambda_j(w)=0$ for some $j$ if and only if $w_0= w_1$.

\begin{prop}
\label{prop:LEsfull}
Let $A^{\mathscr E,\theta,w}$ be given by \eqref{eq:full_cocycle} with $\RR_0^+\times\RR_0^+\ni w=(w_0,w_1)\ne(0,0)$. Let $\mu(w)$ and $\lambda_1(w),\ldots,\lambda_4(w)$ be given by \eqref{eq:mu} and \eqref{eq:lambdafullmodel}. If $w_0\ne w_1$ then $\lambda_j(w)\ne0$ for all $j$, and if $\{\lambda_{i_1}(w),\ldots,\lambda_{i_4}(w)\}$ is a rearrangement such that $|\lambda_{i_1}(w)|\ge\ldots\ge|\lambda_{i_4}(w)|$ then the $j \in \{1,\ldots,4\}$ LEs satisfy
\begin{gather} \label{eq:lowerboundLEs}
 L^j(1/L,A^{\mathscr E,\theta,w}) \ge \sum_{k=1}^j \log \vert \lambda_{i_k}(w)\vert,\\
 L_j(1/L,A^{\mathscr E,\theta,w})\ge \max\bigg\{ \sum_{k=1}^j \log \vert \lambda_{i_k}(w)\vert- (j-1)\log(\vert \mathscr E \vert +2+ \vert w_0\vert+ \vert w_1 \vert),0\bigg\}.
 \label{eq:lowerboundLEfull}
 \end{gather}
In particular, if $(w_0,w_1) = \rho \cdot (\kappa_0,\kappa_1)$ for fixed $(\kappa_0,\kappa_1)\in \mathbb R_0^+ \times  \mathbb R_0^+$ with $\kappa_0\ne\kappa_1$ and if $\rho>0$ is sufficiently large then $ L_j(1/L,A^{\mathscr E,\theta,w})>0$ for $j \in \{1,\ldots,4\}$.
\end{prop}

\begin{proof}
First, since all $\lambda_k$ are non-zero, we find that if $(w_0,w_1) =\rho\cdot(\kappa_1,\kappa_2)$ with $\kappa_i \ge 0$ fixed and $\rho \gg 1$ then $ L_j(1/L,A^{\mathscr E,\theta,w}) >0$ for $j\in\{1,2,3,4\}$ by \eqref{eq:lambdafullmodel} and \eqref{eq:lowerboundLEfull}.

To estimate the Lyapunov exponents, we divide the proof into three cases:
\begin{itemize}
    \item[$(\mathrm{i})$] either $w_0=0$ (the chiral limit),
    \item[$(\mathrm{ii})$] or $w_0\ne0$ and $\mu\ne0$,
    \item[$(\mathrm{iii})$] or $w_0\ne0$ and $\mu=0$.
\end{itemize}
(Note that assumption $(\mathrm{ii})$ contains the anti-chiral limit $w_1=0$ as a special case.)

We first show that under either assumption (i) or (ii) there is a unitary $\mathcal U_\pm(w)$ such that for $\pm\varepsilon$ large we have
\begin{equation}\label{eq:diagonalizationarg}
\mathcal U^{-1}_\pm Q^{\mathscr E,\theta,w}(x+i\varepsilon) \mathcal U_\pm =e^{\pm 2\pi(\varepsilon- i x)} \operatorname{diag}(\lambda_1(w),\ldots,\lambda_4(w))+\mathcal O(1),
\end{equation}
where $\lambda_j(w)\ne0$ for all $1\le j\le 4$.

We begin with case (i), i.e., assume that $w_0=0$. Then $w_1\ne0$ by assumption. For $\pm\varepsilon$ large we have
\begin{align*}
U^+_\mathrm{c}(x+i\varepsilon)&=-\tfrac16 e^{\pm2\pi(\varepsilon-ix)}(1\mp i\sqrt 3)+\mathcal O(1), \quad \pm\varepsilon\gg1,\\
U^-_\mathrm{c}(x+i\varepsilon)&=-\tfrac16 e^{\pm2\pi(\varepsilon-ix)}(1\pm i\sqrt 3)+\mathcal O(1), \quad \pm\varepsilon\gg1.
\end{align*}
With $\nu_{\pm}(\theta):= e^{2\pi i \theta}\frac{1 \pm i\sqrt{3}}2$ and $\pm\varepsilon$ large,
we then find for \eqref{eq:Q} that
\begin{equation*}
\begin{split}
Q^{\mathscr E,\theta,w_1}(x+i\varepsilon)=\frac{w_1e^{\pm2\pi(\varepsilon-ix)}}{3} \begin{pmatrix} 0& 0 &\nu_\mp(-\theta)&  0 \\ 0 & 0&0 & \nu_\pm(\theta) \\ \nu_\pm(-\theta) &  0 & 0 & 0 \\ 0 &\nu_\mp(\theta)& 0 & 0 \end{pmatrix}+\mathcal O(1).
\end{split}
\end{equation*}
To diagonalize the matrix on the right we take 
\begin{equation*}
\mathcal U_\pm  = \begin{pmatrix}0&\bar\omega^{\pm1} & 0& \bar\omega^{\pm1} \\ 
\omega^{\pm1} &0 &\omega^{\pm1} &0\\ 
0& 1 & 0 & -1 \\ 1 & 0 & -1 & 0
\end{pmatrix},\quad \omega=e^{\pi i/3}=\frac{1+i\sqrt 3}2.
\end{equation*}
Then for $\pm\varepsilon$ large we obtain \eqref{eq:diagonalizationarg} with $\lambda_j(w)$ as in \eqref{eq:lambdafullmodel} given by 
\[  \lambda_{1} =  \frac{w_1}3 e^{2\pi i\theta}, \quad \lambda_{2} = \frac{w_1}3 e^{-2\pi i\theta} , \quad
  \lambda_{3} = -\frac{w_1}3 e^{2\pi i\theta} , \quad \lambda_{4} = - \frac{w_1}3 e^{-2\pi i\theta} .
\]

Next, assume that (ii) holds.
For $\pm\varepsilon$ large we have
\begin{align*}
U^+_\mathrm{ac}(x+i\varepsilon)&=U(x+\phi/L+i\varepsilon)=\tfrac13 e^{\pm2\pi(\varepsilon-ix)}e^{\mp 2\pi i\phi/L}+\mathcal O(1), \quad \pm\varepsilon\gg1,\\
U^-_\mathrm{ac}(x+i\varepsilon)&=U(x-\phi/L+i\varepsilon)=\tfrac13 e^{\pm2\pi(\varepsilon-ix)}e^{\pm 2\pi i\phi/L}+\mathcal O(1), \quad \pm\varepsilon\gg1.
\end{align*}
For $\pm\varepsilon\gg1$ we then get from \eqref{eq:Q} that
\begin{equation*}
Q^{\mathscr E,\theta,w}(x+i\varepsilon) = \frac{e^{\pm 2\pi (\varepsilon-ix)}}{3} 
\widetilde Q + \mathcal O(1),   
\end{equation*}
where
\[\widetilde Q = 
\begin{pmatrix} 0 & 0 & w_1\nu_{\mp}(-\theta) & -w_0e^{-2\pi i (\theta\pm\phi/L)} \\ 
0 & 0 &  -w_0e^{2\pi i (\theta\pm\phi/L)} & w_1\nu_{\pm}(\theta) \\ 
w_1\nu_{\pm}(-\theta) &   -w_0e^{-2\pi i (\theta\pm\phi/L)} & 0 & 0 \\ 
  -w_0e^{2\pi i (\theta\pm\phi/L)} & w_1\nu_{\mp}(\theta) & 0 & 0 \end{pmatrix}.
\]
The matrix $e^{\pm 2\pi (\varepsilon-ix)}\widetilde Q/3$
has eigenvalues $e^{\pm 2\pi (\varepsilon-ix)}\lambda_j(w)$ with $\lambda_j(w)$ given by \eqref{eq:lambdafullmodel} and
with corresponding eigenvectors
\[\begin{split} v_{1,\pm} & =\Big(e^{-2\pi i (\theta\pm\phi/L)} (iw_1\sin(2\pi  \theta)- \mu),  w_0 \omega^{\pm 1}, \nu_{\pm}(-\theta\mp\tfrac{\phi}{L})(iw_1 \sin(2\pi  \theta)- \mu),w_0\Big)^\intercal, \\ 
v_{2,\pm} & =\Big(e^{-2\pi i (\theta\pm\phi/L)} (iw_1\sin(2\pi  \theta)+ \mu), w_0 \omega^{\pm 1},  \nu_{\pm}(-\theta\mp\tfrac{\phi}{L}) (iw_1 \sin(2\pi  \theta)+ \mu),w_0\Big)^\intercal, \\
v_{3,\pm} & = \Big(e^{-2\pi i (\theta\pm\phi/L)} (-iw_1\sin(2\pi  \theta)- \mu) , -w_0\omega^{\pm 1},\nu_{\pm}(-\theta\mp\tfrac{\phi}{L}) (iw_1\sin(2\pi  \theta)+\mu) ,w_0\Big)^\intercal,\\
v_{4,\pm} & = \Big(e^{-2\pi i (\theta\pm\phi/L)} (-iw_1\sin(2\pi  \theta)+ \mu) , -w_0\omega^{\pm 1},\nu_{\pm}(-\theta\mp\tfrac{\phi}{L}) (iw_1\sin(2\pi  \theta)- \mu) ,w_0\Big)^\intercal.
\end{split}\]
Note that when $w_0=0$ two of the eigenvectors are zero, and when $\mu=0$ the four eigenvectors collapse to two. However, under assumption (ii) we have $w_0,\mu\ne0$ so we can then define 
\[ \mathcal U_{\pm}=\begin{pmatrix} v_{1,\pm},v_{2,\pm},v_{3,\pm},v_{4,\pm} \end{pmatrix} \]
such that \eqref{eq:diagonalizationarg} holds, where $\lambda_j(w)\ne0$ for all $j$ since $w_0\ne w_1$.

Having established \eqref{eq:diagonalizationarg}, this implies that for $\pm \varepsilon \gg 1$ we have
\[
\begin{pmatrix} \mathcal U_{\pm}^{-1} & 0 \\ 0 & \mathcal U_{\pm}^{-1}
\end{pmatrix} A^{\mathscr E,\theta,w}(x+i\varepsilon) \begin{pmatrix} \mathcal U_{\pm} & 0 \\ 0 & \mathcal U_{\pm} 
\end{pmatrix} = e^{\pm 2\pi(\varepsilon- i x)} \begin{pmatrix} \operatorname{diag}(\lambda_1,\ldots,\lambda_4) +  o(1)  & o(1) I_4  \\  o(1) I_4 & 0_4 \end{pmatrix}.\]
Thus, the cocycle $A^{\mathscr E,\theta,w}(x+i\varepsilon)$ is equivalent to $e^{\pm 2\pi(\varepsilon- i x)} M_{\varepsilon}(x)$, where $M_{\varepsilon}$ is itself an analytic cocycle
\[ M_{\varepsilon} :=\begin{pmatrix} \operatorname{diag}(\lambda_1,\ldots,\lambda_4) +  o(1)  & o(1) I_4  \\  o(1) I_4 & 0_4 \end{pmatrix} \text{ as } \pm \epsilon \to \infty.\]
By the definition \eqref{eq:cocycle} of the $n$-th cocyle $A_n^{\mathscr E,\theta,w}$ and the definition \eqref{eq:LEdef} of Lyapunov exponents, it follows that
\[ L_j(1/L,A^{\mathscr E,\theta,w}(\bullet+i\varepsilon))=2\pi \vert \varepsilon\vert +  L_j(1/L,M_{\varepsilon}).\]
We note that 
\[ M_{\varepsilon} = M + o(1) \text{ with }M=\operatorname{diag}(\lambda_1,\ldots,\lambda_4,0,0,0,0). \]
By continuity of Lyapunov exponents \cite[Theorem 1.5]{AJS15} for analytic cocycles, we then have 
\[ L_j(1/L,M_{\varepsilon}) =  L_j(1/L,M)+o(1).\]
Using then that the right derivatives of the Lyapunov exponents $  L^j(1/L,M_{\varepsilon})$ in $\varepsilon$ are quantized \cite[Theorem 1.4]{AJS15} we conclude, since Lyapunov exponents $  L^j$ are convex and piecewise affine in $\varepsilon$, that for $j \in \{1,2,3,4\}$,
\begin{equation}\label{eq:LEaux1}   L^j(1/L,A^{\mathscr E,\theta,w}(\bullet+i\varepsilon)) = 2j\pi \vert \varepsilon \vert + \sum_{k=1}^j \log \vert \lambda_{i_k}(w)\vert,
\qquad
\vert \varepsilon \vert \ge \varepsilon_0 \gg 1.\end{equation}
Here we use the fact that the rearrangement satisfies $|\lambda_{i_1}(w)|\ge\ldots\ge|\lambda_{i_4}(w)|$.
By convexity, we conclude that \eqref{eq:lowerboundLEs} holds.
This together with \eqref{eq:upper_bds} and Lemma \ref{lem:LEpair} implies \eqref{eq:lowerboundLEfull} for case (i) and (ii).

Now assume instead that (iii) holds. Then $0=\mu=\sqrt{w_0^2-w_1^2\sin^2 2\pi\theta}$ with $w_0\ne0$ and $w_0\ne w_1$, which implies that $\cos 2\pi\theta\ne0$. Using a Jordan decomposition we get as above that 
\[ L_j(1/L,A^{\mathscr E,\theta,w}(\bullet+i\varepsilon))=2\pi \vert \varepsilon\vert +  L_j(1/L,M)+o(1),\]
where now
$$
M=\begin{pmatrix}
    \lambda_1(w) & 1 & 0 & 0\\ 0 &\lambda_1(w) & 0 & 0\\
    0 & 0 & \lambda_3(w) & 1 \\ 0 & 0 & 0 & \lambda_3(w) 
\end{pmatrix},\quad \lambda_j(w)=\pm w_1\cos(2\pi\theta).
$$ 
Using induction is easy to check that 
$$
M^n=\begin{pmatrix}
    A_1 & 0 \\ 0 & A_3\end{pmatrix},\quad 
    A_j=\begin{pmatrix}
        \lambda_j(w)^n & n\lambda_j(w)^{n-1} \\ 0 &\lambda_j(w)^n 
    \end{pmatrix}
$$
which implies that
$$
(M^*)^n M^n=\begin{pmatrix}
    B_1 & 0 \\ 0 & B_3\end{pmatrix},\quad 
    B_j=\begin{pmatrix}
    \lambda_j^{2n} & n\lambda_j^{2n-1} \\ n\lambda_j^{2n-1} &\lambda_j^{2n}+n^2\lambda_j^{2n-2} 
\end{pmatrix}.
$$
Computing the eigenvalues of $B_j$ and using the determinant formula for block matrices we find that the eigenvalues of $(M^*)^n M^n$ are given by
$$
\frac{|\lambda_j|^{2n-2}}2\bigg(n^2+2|\lambda_j|^{2}\pm n^2\sqrt{1+4|\lambda_j|^{2}/n^2}\bigg),\quad j=1,3,
$$
where $|\lambda_1|=|\lambda_3|=\frac{w_1}{3}|\cos 2\pi\theta|$.
Taking the logarithm of the square root, dividing by $n$ and passing to the limit we see that the four Lyapunov exponents $ L_1= L_2= L_3= L_4$ of $M$ are 
\begin{align*}
     L_j(1/L,M)&=\log |\lambda_j(w)|=\log(w_1/3)+\log (|\cos 2\pi\theta|/3).
\end{align*}
We may then repeat the arguments above and coclude that \eqref{eq:LEaux1} and then also \eqref{eq:lowerboundLEs} both hold for case (iii) as well. Thus, \eqref{eq:lowerboundLEfull} holds for all three cases (i)--(iii), i.e., for any $w\ne 0$ with $w_0\ne w_1$.
\end{proof}

\begin{rem}\label{rem:constantEV}
    Proposition \ref{prop:LEsfull} and \eqref{eq:upper_bds} show that for $w_0\neq w_1$  
   \begin{equation}
   \label{eq:as_Lyap}
   L_j(1/L, A^{\mathscr E,\theta,w}) =\log \vert \rho \vert +\mathcal O(1)  
   \end{equation}
     and
    \[L^j(1/L, A^{\mathscr E,\theta,w}) =j\log \vert \rho \vert +\mathcal O(1). \]
    This is true for any $\phi$. When $w_0=w_1$ and $\phi=0$ (the case originally considered in \cite{TM2020}) we don't find a lower bound on the Lyapunov exponents.

    In particular, for $w_0=w_1$ and $\phi=0$ the matrix $V_w^\phi$ has two constant eigenvalues $\pm w_0$ with eigenvectors 
    \[ v_{\pm}(x)=\Big(\mp U(x), \tfrac{\mp(2+ \cos(2\pi x)+\sqrt{3}\sin(2\pi x))}{3}, \tfrac{2 + \cos(2\pi x)+\sqrt{3}\sin(2\pi x)}{3}, U(x) \Big)^\intercal.\]
    This violates the condition in \cite[(1.2b), Theorem $1$]{K17} to establish localization. 
    
    When $w_0=w_1$ and $\phi=1/3$ the matrix $V_w^\phi$ has a constant eigenvalue zero with multiplicity two, while for $\phi=1/4$ the eigenvalues are all non-constant. In general, the eigenvalue expressions are quite heavy even for these special values of $\phi$, and checking whether $V_w^\phi$ has constant eigenvalues for general parameter values is cumbersome. For this reason we prefer to rely on the explicit estimates provided by Proposition \ref{prop:LEsfull} over appealing to \cite[Theorem $1$]{K17}.
\end{rem}

\section{Rational and almost rational moir\'e lengths}
\label{sec:ratio}

In this section we study spectral decomposition of the Hamiltonian $H_w(\theta,\phi,\vartheta)$ by means of its density of states, introduced below. For a more thorough discussion of density of states in a context similar to ours we refer to \cite[\S 5]{kirsch2007invitation}.

We define the regularized trace for $f\in C_c^{\infty}(\RR)$ of the tight-binding model $H=H_w(\theta,\phi,\vartheta)$ in \eqref{eq:Hamiltonian} by 
\begin{equation}\label{eq:regularizedtrace}
    \widehat{\tr}(f(H )) := \lim_{n \to \infty} \frac{ \tr_{\ell^2(\ZZ)}(\indic_{\{-n,\ldots,n\}} f(H))  }{2n+1 }.
\end{equation}
The map $f\mapsto \widehat{\tr}(f(H ))$ is a bounded linear functional on $C_c^0(\mathbb R)$, so by the Riesz representation theorem we have $\widehat{\tr}(f(H ))=\int_{-\infty}^\infty f(x)\,d\nu(x)$ for a measure $\nu$. The measure
$$
\nu(A)=\int_{-\infty}^\infty \indic_A(x)\,d\nu(x),\quad \text{$A$ a Borel set in $\mathbb R$,}
$$
is called the density of states measure of $H$, or just the density of states. The integrated density of states is the cumulative distribution function 
\begin{equation}\label{eq:intDOS}
\nu((-\infty,\lambda])=\int_{-\infty}^\lambda d\nu(x)=\widehat{\tr}(\indic_{(-\infty,\lambda]}(H )).
\end{equation}
We will write $\widehat{\tr}(\indic_{A}(H ))$ for the density of states measure of $H$ of a set $A$, and $\widehat{\tr}(\indic_{(-\infty,\lambda]}(H ))$ for the integrated density of states of $H$.

\subsection{Rational moir\'e lengths}
When $1/L \in \mathbb Q$, the Hamiltonian $H_w$ is a periodic operator. Hence, its spectrum can be studied using Bloch-Floquet theory and we obtain the following spectral decomposition.
\begin{prop}
\label{prop:rational}
The spectrum of the Hamiltonian $H_{w}(\theta,\phi,\vartheta)$ for $1/L \in \mathbb Q$ is purely absolutely continuous and the density of states is continuous. 
\end{prop}
\begin{proof}
We can apply \cite[Theorem $6.10$]{Kuchment} to see that $\sigma_{\operatorname{sc}}(H_w)=\emptyset,$ and any possible point spectrum consists only of flat bands after applying the Floquet transform. Let $H_w(k_x)$ denote the family of Bloch-Floquet operators with quasi-momentum $k_x$ as in \S1.4 in \cite{BW22}.  By \cite[Theorem 6.10 (3)]{Kuchment} the point spectrum of $H_w$ is given by eigenvalues of $H_w(k_x)$ that are constant on an open set $I \ni k_x.$ By Rellich's theorem on the analytic dependence of eigenvalues on a parameter, such eigenvalues however have to be globally constant in $k_x.$

By \cite[Corollary 6.19 (2)]{Kuchment} the integrated density of states is continuous, unless the operators $H_w(k_x)$ exhibit a flat band, i.e., a common eigenvalue that does not depend on $k_x.$ We conclude that the absence of point spectrum for $H_w$ is thus equivalent to the continuity of the integrated density of states.

We will verify the continuity of the integrated density of states in Proposition \ref{prop:DOS_cont} later in this section. 
\end{proof}

\subsection{Absence of flat bands}

To study the density of states through the regularized trace $ \widehat{\tr}(f(H ))$ in \eqref{eq:regularizedtrace}, we associate to $H_w$ a semiclassical pseudodifferential operator on $L^2(\mathbb T)$. The semiclassical parameter is the moir\'e length $h=(2\pi L)^{-1}$, i.e., we are concerned with the limit of large moir\'e lengths $L \gg 1$. We shall use the following version of \cite[Lemma 2.1]{BW22}, the only difference being that in \cite{BW22} the parameters $\phi$ and $\vartheta$ are set to zero, but the same proof works for general $\phi,\vartheta$.

\begin{lemm}\label{lem:PsiDO}
The Hamiltonian $H_w(\theta,\phi,\vartheta)$ in \eqref{eq:Hamiltonian} is unitarily equivalent to the semiclassical operator
$H_{\operatorname{\Psi DO}}(w):L^2(\mathbb T)\to L^2(\mathbb T)$ defined as 
\begin{equation*}
 H_{\operatorname{\Psi DO}}(w)u(x):=(2 t(\theta) \cos(2\pi hD_x)+ t_0  +V_w^\phi(\vartheta+x))u(x).
 \end{equation*}
\end{lemm}

\begin{prop}
The regularized trace for the tight-binding model satisfies for $1/L = \frac{p}{q} \in \mathbb Q$
\[
 \widehat{\tr}(f(H )) = \frac{\sum_{\gamma \in \{0,\ldots,q-1\}} \tr_{\CC^4}( \int_{\mathbb T}\sigma(f(H_{\operatorname{\Psi DO}}))(x,\gamma L^{-1} ) \, d x)}{q}.
 \]
For all other $1/L$
\begin{equation}
\label{eq:irrational}
 \widehat{\tr}(f(H )) =\int_{\mathbb T^2}\tr_{\CC^4}\sigma(f(H_{\operatorname{\Psi DO}}))(x,\xi) \, dx \, d\xi.
 \end{equation}
\end{prop}
\begin{proof}
The formula for the Weyl symbol \cite[Theorem $4.19$]{zworski} implies that if $\mathcal U$ is the unitary map defining the equivalence in Lemma \ref{lem:PsiDO}, then
\[
\begin{split}
 \widehat{\tr}(f(H )) &= \lim_{n \to \infty} \frac{ \tr_{\ell^2(\ZZ)}(\indic_{\{-n,\ldots,n\}} f(H))  }{2n+1} \\
  &= \lim_{n \to \infty} \frac{\tr_{\ell^2(\ZZ)}(\indic_{\{-n,\ldots,n\}}  \mathcal U^{-1}  f(H_{\operatorname{\Psi DO}}) \mathcal U)}{2n+1} \\
    &= \lim_{n \to \infty} \frac{\sum_{\gamma \in \{-n,\ldots,n\}^2} \tr_{\CC^4}( \int_{\mathbb T} e^{-i  \gamma x} f(H_{\operatorname{\Psi DO}}( x,h D_x,\theta))  e^{i  \gamma x } \, d x)}{2n+1} \\
    &= \lim_{n \to \infty} \frac{\sum_{\gamma \in \{-n,\ldots,n\}^2} \tr_{\CC^4}( \int_{\mathbb T}\sigma(f(H_{\operatorname{\Psi DO}}))(x,\gamma L^{-1}) \, d x)}{2n+1 }.
 \end{split}
\]
When $1/L$ is rational
\[\gamma \mapsto \tr_{\CC^4}\left( \int_{\mathbb T}\sigma(f(H_{\operatorname{\Psi DO}}))(x,\gamma L^{-1}) \, d x\right)\] is periodic and thus we obtain
\[ \widehat{\tr}(f(H )) = \sum_{\gamma \in \{0,\ldots,q-1\}} \tr_{\CC^4} \int_{\mathbb T}\sigma(f(H_{\operatorname{\Psi DO}}))(x,\gamma L^{-1}) \, dx.\]

If $1/L$ does not satisfy this rationality condition, then the translation $(T^{\gamma}u)(x) = u(x+1/L)$ is a \emph{uniquely} ergodic endomorphism on the probability space $\RR/\ZZ$ and therefore using the continuity of the Weyl symbol, it follows from \cite[Theorem $6.19$]{W82} that
\[
\widehat{\tr}(f(H )) =\int_{\RR^2/\ZZ^2}\tr_{\CC^4}\sigma(f(H_{\operatorname{\Psi DO}}))(x,\xi) \, dx \, d\xi.
\qedhere\]
\end{proof}
In the sequel, we shall write
\[ \widehat{\tr}_{\Omega}(\operatorname{Op}(a)) :=   \int_{\mathbb T^2}\tr a(x,\xi ) \, d x \, d\xi.\]
To see that the density of states for commensurable angles coincides with formula \eqref{eq:irrational}, we use the following lemma  which we actually state for one-dimension Schrödinger operators, but whose proof carries immediately over to arbitrary dimensions and matrix-valued operators, whose kinetic and potential operators are sums and products of exponential functions, including our operator of interest.
\begin{lemm}
Let $S: \ell^2(\ZZ) \rightarrow \ell^2(\ZZ)$ be a discrete Schrödinger operator with a potential that has a finite Fourier representation 
\[ Su_n = u_{n+1} + u_{n-1} +  \sum_{j=-m}^m a_j e^{2\pi i j n/L},\]
then $S$ is unitarily equivalent to a pseudodifferential operator 
\[
S_{\operatorname{\Psi DO}}f(x) = 2 \cos(2\pi x)f(x)+  \sum_{j=-m}^m a_j e^{2\pi i j/L D} f(x),
\]
with $D=-i \partial_x$,
and its density of states, defined by the regularized trace $ \widehat{\tr}(f(S))$ given by \eqref{eq:regularizedtrace}, satisfies
\[ \widehat{\tr}(f(S))= \int_{\mathbb T^2} \sigma(f(S_{\operatorname{\Psi DO}}))(x, \xi) \, dx \, d\xi.\]
\end{lemm}
\begin{proof}
We consider the one-dimensional operator $S: \ell^2(\ZZ) \rightarrow \ell^2(\ZZ)$
\[ Su_n = u_{n+1} + u_{n-1} +  \sum_{j=-m}^m a_j e^{2\pi i j n/L}\]
where we assume that $a_j = \bar{a}_{-j}.$ Then this operator is equivalent to a pseudodifferential operator on $\mathbb T$ given as 
\[
\begin{split}
S_{\operatorname{\Psi DO}}f(x) &= 2 \cos(2\pi x)f(x)+  \sum_{j=-m}^m a_j e^{i j/LD} f(x)\\
&=2 \cos(2\pi x)f(x)+  \sum_{j=-m}^m a_j f(x+2\pi j/L).
 \end{split}
\]
On the level of the symbol of the operator, the commensurable and incommensurable expressions for the integrated density of states always coincide due to 
$\sum_{k=0}^{n-1} e^{\frac{2\pi i k}{n} } = 0.$
Similar reasoning and the composition formula for symbols of operators implies that the two formulas coincide for $f$ in the functional calculus being any polynomial. Thus, Weierstrass's theorem implies that the two formulations coincide for any continuous function and since the map from operators defined in the symbol class $\mathcal S(1)$ to their Weyl symbols is continuous under uniform convergence, the result follows.
\end{proof}

\begin{prop}
\label{prop:DOS_cont}
The integrated density of states of the tight-binding Hamiltonian $H$ is a continuous function. In particular, the Hamiltonian does not possess any flat bands at commensurable length scales $L \in \mathbb Q^+$.
\end{prop}
\begin{proof}
Since the integrated density of states \eqref{eq:intDOS} is a cumulative distribution function, it is a c\`adl\`ag function (continuous from the right, with limits from the left).
To show that it is continuous it therefore suffices to show that the density of states measure has no atoms, that is, $\widehat{\tr}(\indic_{\{\lambda\}}(H))=0 $ for all $\lambda \in \RR$ fixed. (Compare e.g., with \cite[Theorem 5.14]{kirsch2007invitation}.)
Since translations only appear at leading order in $H$, one then observes that a solution $H \psi =\lambda \psi$ is uniquely determined inside $\{-n,-n+1,\ldots,n\}$ by specifying it on $\{\pm n \}.$  Since these are only two points, we find that 
\[ \tr_{\ell^2(\ZZ)}(\indic_{\{-n,-n+1,\ldots,n\}} \indic_{\{\lambda\}}(H))=  \mathcal{O}(1)\]
and hence $\widehat{\tr}(\indic_{\{\lambda\}}(H))=0.$
Thus, there cannot be any flat band at $\lambda$, as this would imply that $\widehat{\tr}(\indic_{\{\lambda\}}(H))>0$ by \cite[Corollary 6.19 (2)]{Kuchment}.
\end{proof}

To summarize, we have thus shown that for rational $1/L,$ the spectrum is purely absolutely continuous. We shall now turn to almost rational moir\'e lengths in the next subsection. 
\subsection{Almost-rational (Liouville) moir\'e lengths}
\label{sec:Gordon}
Recall that a number $\alpha \in \RR^{+} \backslash \mathbb Q$ is called a \emph{Liouville number} if for all $k \in \mathbb N$ there are $p_{k},q_k \in \mathbb N$ such that 
\begin{equation}\label{eq:Liouville}
\vert \alpha-p_{k}/q_{k} \vert <q_k^{-1} k^{-q_{k}}.    
\end{equation}
For such numbers, it is well-known that quasi-periodic Schrödinger operators do not exhibit point spectrum. As the next proposition shows, this holds true for our matrix-valued discrete operators, when $1/L$ is Liouville. 
\begin{prop}
\label{prop:Liouville}
Let $1/L$ be a Liouville number, then the Hamiltonian $H_w(\theta,\phi,\vartheta)$ does not have any point spectrum. In particular, if in addition $ L_4(1/L, A^{\mathscr E,\theta,w})>0$\footnote{This holds e.g.~under the assumptions of Proposition  \ref{prop:LEsfull}.}, then the spectrum of the Hamiltonian $H_w(\theta,\phi,\vartheta)$ is purely singular continuous.
\end{prop}
 \begin{proof}
Given $k\in\mathbb N$, let $p_k,q_k$ be approximates of $1/L$ as in \eqref{eq:Liouville}. Then a simple calculation shows that for some universal $C>0$ we have the operator norm estimate
\begin{equation}\label{eq:firstbound}
\Big\Vert A^{\mathscr E,\theta,w}(\vartheta+i/L)-A^{\mathscr E,\theta,w}(\vartheta+ip_k/q_k) \Big\Vert \le C \vert w \vert   \vert i/L-ip_{k}/q_{k} \vert\le 8C\vert w \vert  k^{-q_k}
 \end{equation}
 for any $i$ such that $|i|\le 8q_k$. Writing 
 $$
 S_n=\bigg\lVert \prod_{i=n}^1 A^{\mathscr E,\theta,w}(x+i/L)-\prod_{i=n}^1 A^{\mathscr E,\theta,w}(x+ip_k/q_k) \bigg\rVert
 $$
 we get from \eqref{eq:firstbound} together with \eqref{eq:normboundcocycle} that if $|n|\le 8q_k$ then
 $$S_n\le e^{C'(n-1)}8C\vert w \vert  k^{-q_k}+e^{C'} S_{n-1},$$
 where $C'=\log(2+|\mathscr E|+|w_0|+|w_1|)$. Since $|n-1|\le 8q_k$, an induction argument gives
 \begin{align*}S_n 
 &\le e^{C'(n-1)}8C|w|k^{-q_k}+e^{C'}(e^{C'(n-2)}8C|w|k^{-q_k}+e^{C'}S_{n-2})\\&=2e^{C'(n-1)}8C|w|k^{-q_k}+e^{2C'}S_{n-2}\\&\le\ldots\le (n-1)e^{C'(n-1)}8C|w|k^{-q_k}+e^{(n-1)C'}S_{1} \le ne^{C'(n-1)}8C|w|k^{-q_k}. \end{align*} Since $|w|\le e^{C'}$ we get $S_n\le n e^{C'n}8Ck^{-q_k}$ for all $|n|\le 8q_k$, which implies that
\begin{equation}\label{eq:Snestimate}
\begin{split}
 &\sup_{\vert n \vert \le 8 q_{k}}\left\lVert \prod_{i=n}^1 A^{\mathscr E,\theta,w}(x+i/L)-\prod_{i=n}^1 A^{\mathscr E,\theta,w}(x+ip_k/q_k) \right\rVert   \le\sup_{\vert n \vert \le 8q_k} 8C\vert n \vert e^{C' \vert n \vert} k^{-q_k}.
\end{split}
\end{equation}

Recall then that if $T \in \operatorname{GL}(8,\CC)$ then by the Cayley-Hamilton theorem there are $c_i$ not all zero such that $\sum_{i=0}^8 c_i T^i =0.$
Normalizing, we can assume that one of the $c_i$ is equal to one  and the other ones, $c_j$ with $j \neq i$, are at most one in absolute value, i.e. $\vert c_j \vert \le \vert c_i\vert.$
Hence, we find from 
\[v = -\sum_{j\neq i} c_j T^{j-i}v\]
for $v$ normalized that $\max_{j \neq i} \{ \Vert T^{j-i} v\Vert\} \ge 1/8$ with maximum taken over $j\in\{0,\ldots,8\}$ such that $j\ne i$.
This shows that for any normalized vector $v$ we have 
\begin{equation}
\label{eq:1/8}
 \max (\Vert T^{-8}v \Vert,\ldots,\Vert T^{-1} v\Vert,\Vert T^{1} v \Vert,\ldots,\Vert T^{8}v \Vert) \ge \max_{j \neq i} \{ \Vert T^{j-i} v\Vert\} \ge  1/8.
 \end{equation}
Applying this result, we find writing $[k] = \{1,2,\ldots,k\}$ that
\[ \begin{split}
\max_{n \in [8q_k]}&\left\{\Vert (\psi_n,\psi_{n-1}) \Vert, \Vert (\psi_{-n+1},\psi_{-n}) \Vert\right\} \\
=\max_{n \in [8q_k]}&\left\{\left\Vert \prod_{i=n}^1 A^{\mathscr E,\theta,w}(x+ip_k/q_k)(\psi_1,\psi_0)^\intercal \right\rVert, \left\Vert \left(\prod_{i=n}^1 A^{\mathscr E,\theta,w}(x+ip_k/q_k)\right)^{-1}(\psi_1,\psi_0)^\intercal \right\rVert\right\} \\
&\ge \tfrac{1}{8}\Vert (\psi_1,  \psi_0) \Vert.
\end{split}\]

Here we used \eqref{eq:1/8} applied to $T=\prod_{i=q_k}^1 A^{\mathscr E,\theta,w}(x+ip_k/q_k)$ since $A^{\mathscr E,\theta,w}(x+ip_k/q_k)$ is a $q_k$ periodic matrix in $i$.
Together with \eqref{eq:Snestimate} this gives
\[ \limsup_{\pm n \to \infty} \frac{\Vert (\psi_{n+1}, \psi_{n})\Vert }{\Vert (\psi_1, \psi_{0})\Vert} \ge \frac{1}{8}.\]

This shows that if $1/L$ is a Liouville number, the operator does not have any point-spectrum. The absence of absolutely continuous spectrum in case of positive LEs follows immediately from Kotani-Simon theory \eqref{eq:Kotani-Simon}. 
\end{proof}

\section{Diophantine moir\'e lengths and Anderson localization}
\label{sec:AL}
We saw in the previous section that for rational numbers and irrational moir\'e length scales $1/L$ that are close to rational ones (Liouville numbers), the Hamiltonian does not exhibit any point spectrum. We will now focus on moir\'e lengths $1/L$ described by real numbers on the opposite end, satisfying for some $t>0$ a diophantine condition, $1/L \in \operatorname{DC}_t$.

 Recall that $\alpha \in\mathbb{R}^d$ is called {\it Diophantine} if there is a $t>0$ such that $\alpha \in {\operatorname{DC}}_t^d$, where
\begin{equation}\label{dio}
{\operatorname{DC}}^d_t:=\left\{\alpha \in\mathbb{R}^d:  \inf_{j \in \mathbb{Z}}\left|\langle n,\alpha  \rangle - j \right|
> \frac{t}{|n|^{d+1}},\quad \forall \, n\in\mathbb{Z}^d\backslash\{0\} \right\},
\end{equation}
where $\langle n,\alpha  \rangle $ is the Euclidean inner product.
In case that $d=1$ we will usually just drop the superscript.

We present one method in Subsection \ref{sec:Dioph} originally due to Bourgain \cite[Chapter $10$]{B}, which applies to the Hamiltonian \eqref{eq:Hamiltonian} in a very general sense and one more refined approach  in Subsection \ref{sec:AVAL} for a modified version of the anti-chiral Hamiltonian that goes back to ideas of Jitomirskaya \cite{J99}. The latter approach has been originally introduced for the almost Mathieu operator and applies to a larger range of moir\'e length scales. The first approach has been extended to the matrix-valued setting by Klein \cite{K17}. Both approaches also imply dynamical localization
\begin{equation}
    \label{eq:dyn_loc}
 \sup_{t>0} \left( \sum_{n \in \mathbb Z} (1+n^2) \Big\vert (e^{-itH} \psi)(n) \Big\vert^2 \right)^{1/2} < \infty,
 \end{equation}
 where $e^{-itH}$ is the time-evolution operator of the quantum dynamics,
as explained in \cite[Chapter 10]{B},  see Figure \ref{fig:localization}. 
The above diophantine condition \eqref{dio}, which appears naturally in the localization proofs in \cite{B,K17}, applies to a set of full measure of real numbers.

\begin{figure}
\includegraphics[width=6cm]{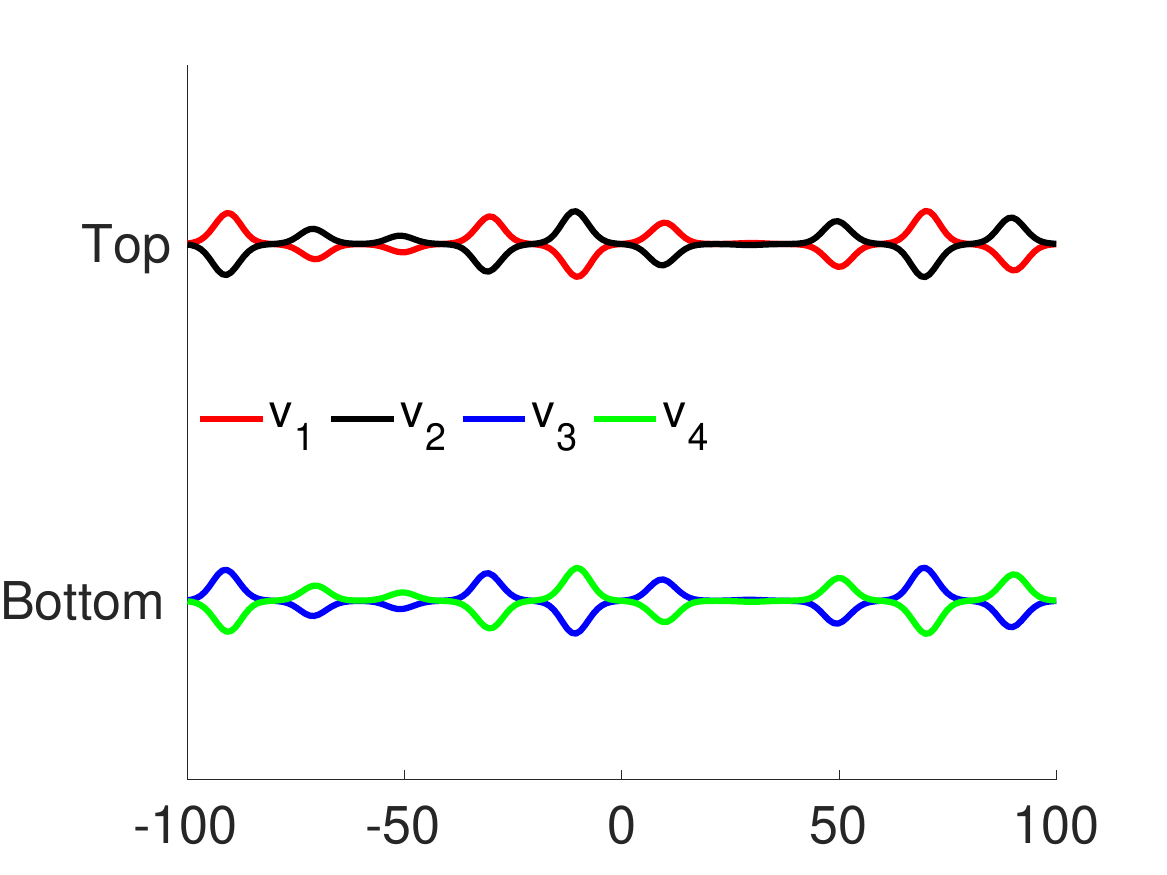} \includegraphics[width=6cm]{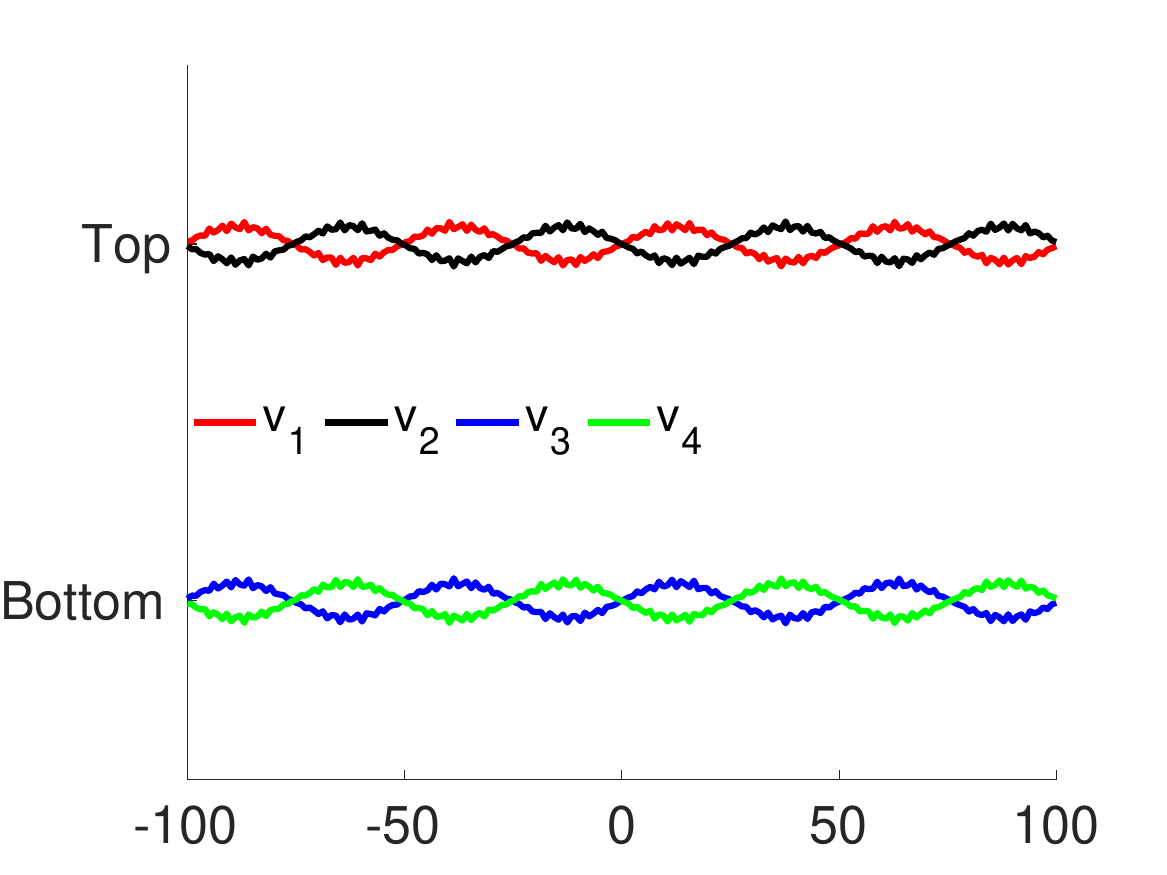} \\
\includegraphics[width=6cm]{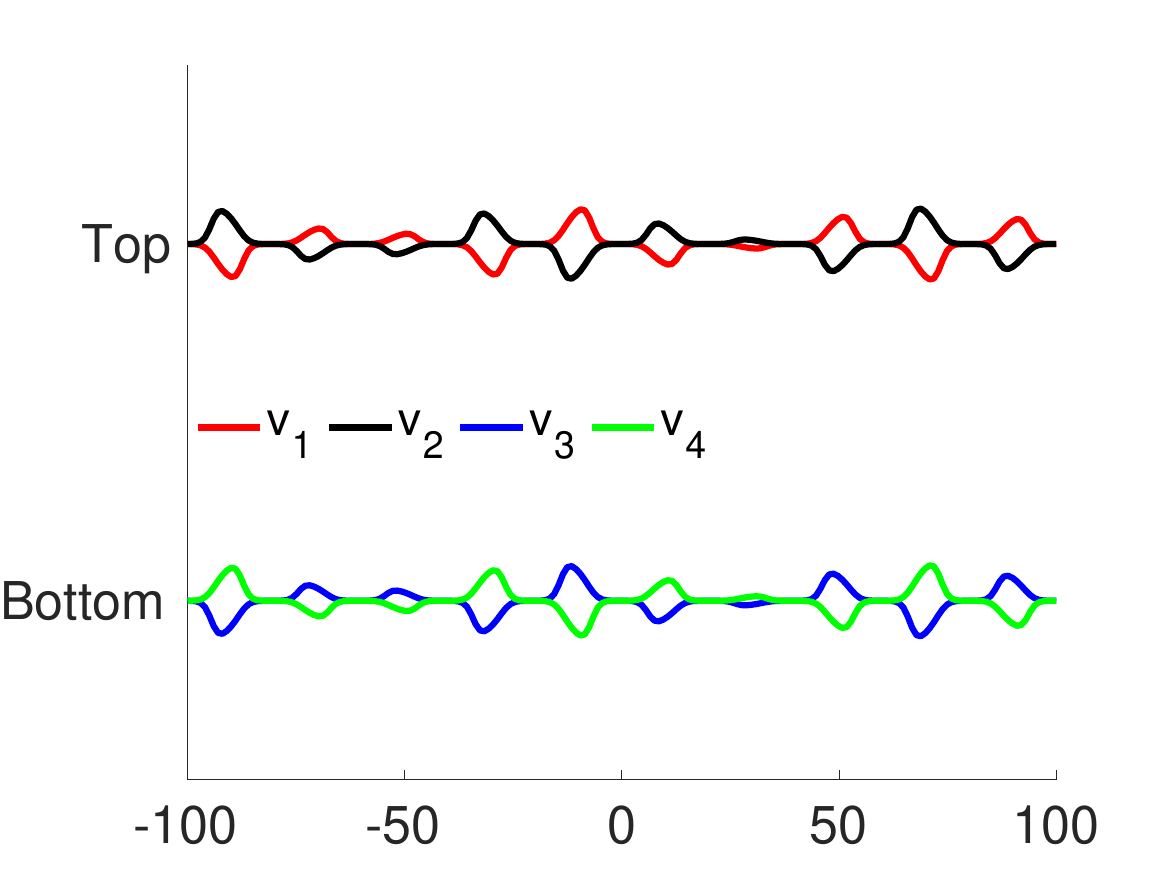} \includegraphics[width=6cm]{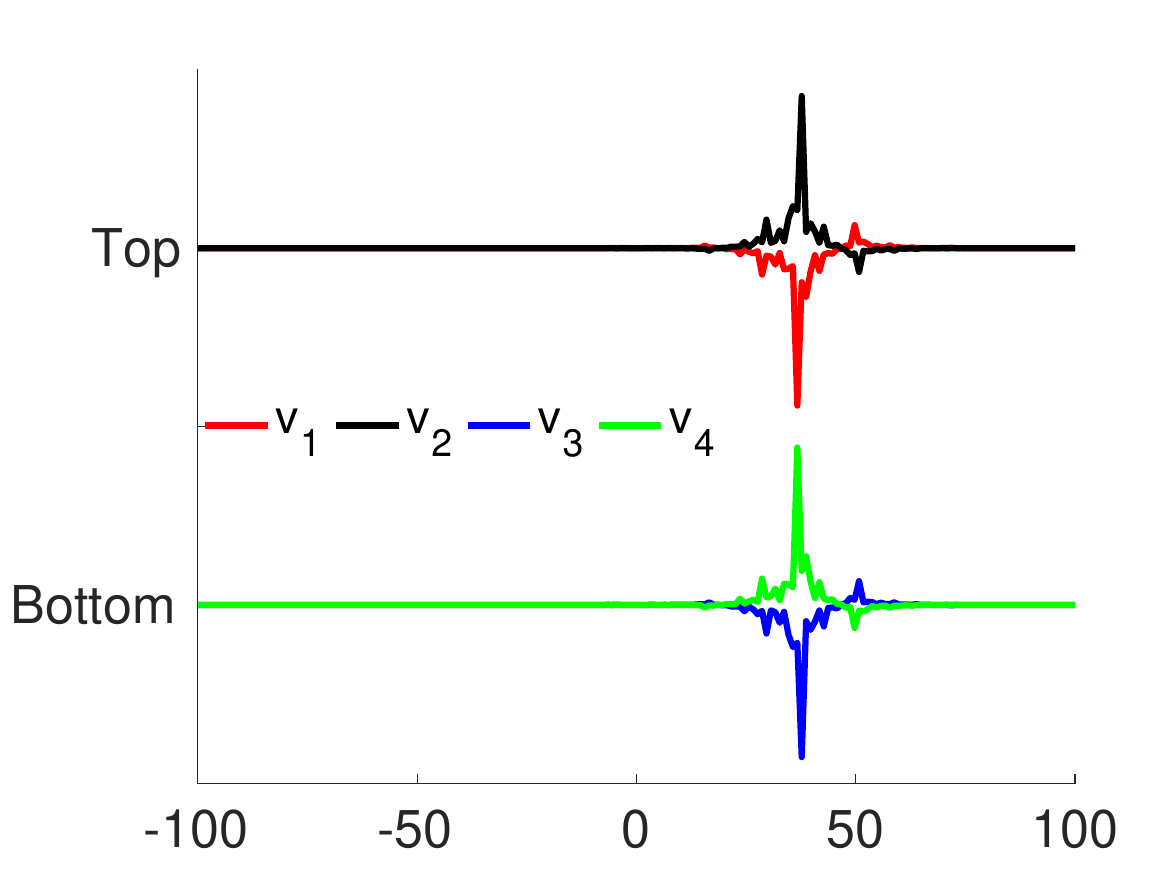}
\caption{\label{fig:localization}Lowest eigenfunction of chiral Hamiltonian restricted to interval $\{-100,-99,\ldots,100\}$. Figures on the left are for \emph{rational} length scales $L=20,$ whereas on the right we study the strongly \emph{irrational} (diophantine) $L=1/\text{golden mean}.$ The top figures correspond to $w_0= \frac{3}{2}$, the bottom ones to $w_0=6.$}
\end{figure}

To start, we shall now recall the definition of \emph{generalized eigenfunctions} which characterize the spectrum by the Schnol-Simon theorem \cite[Theorem $2.9$]{CFKS} whose proof can be easily adapted to matrix-valued Hamiltonian. See also \cite[Theorem 4.1]{jm} for a clear statement.
\begin{defi}\label{def:genEV}
$\mathscr E \in \mathbb{R}$ is called a \emph{generalized eigenvalue}  if there is a \emph{generalized eigenfunction} $u: \ZZ \to \CC$ with $\vert u_n \vert \lesssim \langle n \rangle$ satisfying $H_w(\theta,\phi,\vartheta) u = \mathscr E u.$
\end{defi}
To establish Anderson localization it therefore suffices to show that any generalized eigenfunction decays exponentially and the usual proof proceeds via Green's function estimates showing exponential decay for the off-diagonal entries of the Green's function. 
More precisely, one aims to show that 
$ \vert u_n \vert \le e^{-c\vert n \vert} \text{ as } \vert n \vert \to \infty.$
Our main result of this section is
\begin{theo}
\label{theo:Aloc1}
 Let the coupling parameter $w$ in the Hamiltonian $H_w(\theta,\phi,\vartheta)$ be given by $w=(w_0,w_1) = \rho \cdot(\kappa_0,\kappa_1)$ for fixed $(\kappa_0,\kappa_1) \in \mathbb R_0^+ \times  \mathbb R_0^+$ with $\kappa_0\ne \kappa_1$. Then, for fixed $\vartheta\in\mathbb{T}$ and  $\rho >0$ sufficiently large the Hamiltonian $H_w(\theta,\phi,\vartheta)$ satisfies for a conull set of reciprocal length scales $L \in \mathbb R^+$ Anderson localization, i.e., pure point spectrum with exponentially decaying eigenfunctions and also dynamical localization \eqref{eq:dyn_loc}.
\end{theo}

\subsection{Preliminaries}
In the sequel, we shall write for a block matrix $A \in \mathbb C^{4n \times 4n}$ where $n \in \mathbb N$ 
$$A = (A_{\gamma,\gamma'})_{\gamma,\gamma' \in [4n]} = (A(i,j))_{i,j \in [n]}$$
with $A(i,j) \in \CC^{4 \times 4}$ being themselves matrices, whereas $A_{\gamma,\gamma'}$ are scalars. Recall that $[n]:=\{1,\ldots,n\}.$

Let $\mathcal N=[n_1,n_2] = \{n \in \mathbb Z: n_1 \le n \le n_2\},$ and define two canonical restrictions
\[
\begin{split}
P_{\mathcal N}^- &= \begin{pmatrix}  0_{\ell^2((-\infty,n_1-1];\CC^4)} & I_{\ell^2(\mathcal N;\CC^4)} &  0_{\ell^2([n_2+1,\infty);\CC^4)} \end{pmatrix} \text{ and } \\
P_{\mathcal N}^{+} &= \begin{pmatrix}  0_{\ell^2((-\infty,n_1-2];\CC^4)}  & 0_2\oplus I_{2} &I_{\ell^2([n_1,n_2-1];\CC^4)} &  I_2 \oplus 0_{2} & 0_{\ell^2([n_2+1,\infty);\CC^4)}  \end{pmatrix}.
\end{split}
\]
 Thus, $P_{\mathcal N}^{+}$ is shifted by two components compared to $P_{\mathcal N}^-.$ The operator $P_{\mathcal N}^- $ is just the projection onto $\CC^4$-valued elements on $\mathcal N.$ On the other hand, $P_{\mathcal N}^+ $ projects onto $\CC^4$-valued elements on $[n_1,n_2-1]$ and in addition the last two components at $n_1-1$ and the first two components of $n_2.$ 

We shall omit the $-$ sign in the following, when there is no need to work with both $+$ and $-$.
In the case that $\mathcal N=[1,N]$ we also write $[N]$ instead of $\mathcal N$ and introduce
\begin{equation}\label{pro+-}
H_w^{\pm, \mathcal N}(\theta,\phi,\vartheta):= P^{\pm}_{\mathcal N} H_w(\theta,\phi,\vartheta) P^{\pm}_{\mathcal N}.
\end{equation}
The Hamiltonian $H_w^{-, \mathcal N}$ is obviously defined on $\ell^2(\mathcal N;\CC^4).$ The Hamiltonian $H_w^{+, \mathcal N}$ is defined on $[n_1,n_2-1]$, but in addition takes the two last components of the point $n_1-1$ and the first two components at $n_2$ into account. Thus, by shifting two components, we can (and shall) also consider this one as an operator on $\ell^2(\mathcal N;\CC^4).$

Now, let $\mathscr E \notin \Spec( H_w^{\pm, \mathcal N}(\theta,\phi,\vartheta))$ and $n,m \in \mathcal N$, and define the \emph{Green's function} $ G^{\pm, \mathcal N}_w(\theta,\phi,\vartheta,\mathscr E) \in \CC^{4 \vert \mathcal N \vert \times 4 \vert \mathcal N \vert}$ by
\[  G^{\pm, \mathcal N}_w(\theta,\phi,\vartheta,\mathscr E)(n,m):=( H_w^{\pm,\mathcal N}(\theta,\phi,\vartheta)-\mathscr E)^{-1}(n,m).\]
The Green's function is a $\vert \mathcal N\vert \times \vert \mathcal N\vert$ block matrix, with blocks that are themselves $4\times 4$ matrices over $\CC.$

Let $\varphi$ be a solution to $ (H_w(\theta,\phi,\vartheta)-\mathscr E)\varphi=0$ with $\mathscr E \notin \Spec( H_{w}^{\pm, \mathcal N}(\theta,\phi,\vartheta))$, then it follows as for discrete Schrödinger operators, that for $n$ located between $n_1$ and $n_2$, we have
\[ \varphi_n = - G^{\pm, \mathcal N}_w(\theta,\phi,\vartheta,\mathscr E)(n,n_1)t(\theta)\varphi_{n_1-1} -  G^{\pm, \mathcal N}_w(\theta,\phi,\vartheta,\mathscr E)(n,n_2)t(\theta)\varphi_{n_2+1}.\]
Since $\Vert t(\theta)\Vert=1$ we then have
\begin{equation}
\label{eq:GFestm}
 \Vert \varphi_n\Vert \le \Vert  G^{\pm, \mathcal N}_w(\theta,\phi,\vartheta,\mathscr E)(n,n_1) \Vert \Vert \varphi_{n_1-1}\Vert+ \Vert  G^{\pm, \mathcal N}_w(\theta,\phi,\vartheta,\mathscr E)(n,n_2) \Vert \Vert \varphi_{n_2+1}\Vert .
 \end{equation}
This identity together with good decay estimates on the Green's function will imply the decay of generalized eigenfunctions of the Hamiltonian.

In terms of $\mathscr V_i(w,\phi,\vartheta,\mathscr  E) = t_0 + V_w(\vartheta+\tfrac{i}{L},\phi) - \mathscr E, $ we can write
\[ H_w^{[N]}(\theta,\phi,\vartheta) - \mathscr E = \begin{pmatrix} \mathscr V_1(w,\phi,\vartheta,\mathscr E) & t(\theta) & 0& \cdots & 0\\ 
t(\theta) &\mathscr V_2(w,\phi,\vartheta,\mathscr  E) & t(\theta) &  \ddots &\vdots \\
0 & t(\theta) &\mathscr V_3(w,\phi,\vartheta,\mathscr E)& \ddots  & 0 \\
\vdots & \ddots & \ddots & \ddots & t(\theta)\\
0 & \cdots & 0& t(\theta)  &   \mathscr V_n(w,\phi,\vartheta,\mathscr E)
\end{pmatrix}.\]
Each entry of this block matrix is a $4 \times 4$ matrix and $H_w^{[N]}$ is a matrix of size $4N \times 4N.$ As mentioned above, we write $H_w^{[N]}$ instead of $H_w^{-,[N]}=H_w^{-,\mathcal N}$ for brevity.

\subsection{Almost sure Anderson localization}
\label{sec:Dioph}

We start by introducing the minors
\[
 \mu^{[N]}_{(\alpha,\alpha')}(\theta,\phi,\vartheta,w,\mathscr E):=\operatorname{det}(((H_w^{[N]}(\theta,\phi,\vartheta)-\mathscr E)_{\beta,\beta'})_{\beta\ne\alpha,\beta'\ne\alpha'}).
 \]
The importance of the minors is due to Cramer's rule which in our case, for a self-adjoint real-valued matrix reads
\begin{equation}\label{eq:Cramer}
\begin{split}
(H_w^{[N]}(\theta,\phi,\vartheta)-\mathscr E)^{-1}_{\alpha,\alpha'} = \frac{(-1)^{\alpha+\alpha'}\mu^{[N]}_{(\alpha,\alpha')} (\theta,\phi,\vartheta,w,\mathscr E)}{\operatorname{det}(H_w^{[N]}(\theta,\phi,\vartheta)-\mathscr E)}.
\end{split}
\end{equation}

For $\gamma \in [4n]$ we introduce $n(\gamma) \in \mathbb N$ such that $\gamma = 4n(\gamma)+r \text{ with }-4 \le r < 0.$
Then, by \cite[Proposition $2$]{K17} there is $C<\infty$ such that for all $N$, all $\alpha, \alpha' \in [4N]$, $|\mathscr E|\leq C\rho$ and all $\vartheta,\phi,\theta,$ we have
\begin{equation}
    \label{eq:Klein_bound}
 \frac{1}{4N} \log\big\vert \mu^{[N]}_{(\alpha,\alpha')}(\theta,\phi,\vartheta,w,\mathscr E)\big\vert \le \left(1 - \frac{\vert n(\alpha)-n(\alpha')\vert}{4N}\right)\log\rho +C,
 \end{equation}
where $\rho$ is as in Proposition \ref{prop:LEsfull}.

Next, we turn our attention to 
deviations of the ergodic mean in the Thouless formula.
In the following, let \[u_N(x):=u_N(\theta,\phi,\mathscr E,x):=\frac{\log {\lvert\det(H_w^{[N]}(\theta,\phi,x)-\mathscr E)\rvert} }{4N}\] and define the set $\mathcal B_N^M(1/L) :=\mathcal B_N^M(1/L,\theta,\phi,\mathscr E)$ by
\[ \mathcal B_N^M(1/L) :=
\left\{\vartheta \in \mathbb T: \frac{1}{M} \sum_{j=0}^{M-1} u_N(\vartheta+j/L) \le (1-\delta)  \log\rho \right\}\]
for fixed $\mathscr E$ and $1/L \in \operatorname{DC}_t,$ with $\operatorname{DC}_t$ defined in \eqref{dio}. We then have the following proposition, similar to \cite[Proposition 7.19]{B} and \cite[Proposition $4$]{K17}, which shows that the set of bad frequencies at which the Green's function has no good a priori decay properties is small.
\begin{prop}
\label{prop:good_decay}
Fix $t>0$ and let $1/L \in \operatorname{DC}_t.$ For any $N,M$ large enough, the set $\mathcal B_N^M(1/L)$ is exponentially small in M, such that there is $a>0$ such that $\vert \mathcal B_N^M(1/L) \vert <e^{-M^a}$ and is a semi-algebraic set\footnote{See \cite[Chapter $9$]{B} for a comprehensive definition of this concept.} of degree $\mathcal O(N^2M).$ 
Moreover, there is $p \in \mathbb N$ such that for all $\vartheta$
\[ \left\lvert\left\{ 0 \le n <N^p: (H_w^{[N]}(\vartheta+n/L)-\mathscr E)^{-1} \text{ is not a good Green's function}\,\right\}\right\rvert \ll N^p\]
where the Green's function is called \emph{good}, if for some $\varepsilon>0$ we have
 \[\frac{\log |(H_w^{[N]}(\vartheta+n/L)-\mathscr E)^{-1}_{\alpha,\alpha'}|}{N} <-\left(\frac{\vert n(\alpha)-n(\alpha')\vert}{N}-\varepsilon \right)   \log\rho.\]
\end{prop}
\begin{proof}
The diophantine condition \eqref{dio} enters in the following quantitative version of Birkhoff's ergodic theorem, see \cite[Theorem $6.5$]{DK16}, \cite[p.~17]{K17}, where we have used Proposition \ref{prop:LEsfull} to estimate the Lyapunov exponents as in \eqref{eq:as_Lyap} of Remark \ref{rem:constantEV}: 
Let $1/L \in \operatorname{DC}_t$ and $M \ge t^{-2}$, then for $S = \mathcal O(\log\rho)$ there is $a>0$ such that
\begin{equation}
\label{eq:ergodic_close}
 \left\lvert \left\{\vartheta \in \mathbb T: \left\lvert \frac{1}{M} \sum_{j=0}^{M-1} u_N(\vartheta+j/L) - \int_{\mathbb T} u_N( t)  \, dt  \right\rvert > SM^{-a} \right\} \right\rvert <  e^{-M^a}.
 \end{equation}
 In other words, the set of points where the ergodic average is far away from the average is exponentially small. 
If $\vartheta$ is not in the set in \eqref{eq:ergodic_close}, then by \cite[Proposition $3$]{K17} and \eqref{eq:as_Lyap}
\begin{equation}
\begin{split}
\label{eq:good_green_function}
 \frac{1}{M} \sum_{j=0}^{M-1} u_N(\vartheta+j/L) &\ge  \int_{\mathbb T} u_N(t)  \, dt-SM^{-a} =  \log\rho (1-\mathcal O(1)M^{-a}) - \mathcal O(1)\\
 &\ge (1-\delta)  \log\rho, 
\end{split}
\end{equation}
where $\delta$ can be chosen arbitrarily small for $M$ and $N$ large enough. This implies the estimate on the measure of $B_N^M(1/L,\theta,\phi,\mathscr E).$

One can then estimate the Green's function by Cramer's rule \eqref{eq:Cramer}.
Thus,
\[ \frac{\log| (H_w^{[N]}(\theta,\phi,\vartheta)-\mathscr E)^{-1}_{\alpha,\alpha'}|}{4N} = \frac{\log|\mu^{[N]}_{(\alpha,\alpha')}(\theta,\phi,\vartheta,w,\mathscr E)|}{4N}- u_N(\theta,\phi,\vartheta).\]

Let $\vartheta \notin \mathcal B_N^M(1/L),$ then we have for some $j \in \{0,\ldots,M-1\}$ that 
\[ u_N(\vartheta+j/L) > (1-\delta)\log\rho.\]
This implies that for this choice of $j$, we have together with {\eqref{eq:Klein_bound} and \cite[Proposition $2$]{K17}}  for some $C>0$ and $\rho>0$ large enough
\begin{equation}
    \begin{split}
    \label{eq:good_estm}
        \frac{\log| (H_w^{[N]}(\theta,\phi,\vartheta+j/L)-\mathscr E)^{-1}_{\alpha,\alpha'}|}{4N} &\le\log\rho\Big(1-\tfrac{\vert n(\alpha)-n(\alpha')\vert}{4N}\Big) 
        +C- (1-\delta)   \log\rho\\
        &<-\left(\tfrac{\vert n(\alpha)-n(\alpha')\vert}{4N}-\varepsilon \right)  \log\rho
    \end{split}
\end{equation}
for some $\varepsilon$ sufficiently small. This shows that the Green's function satisfies an exponential decay estimate.
We now observe that for $\vartheta \in \mathcal B_N^M(1/L)$,
\[
\begin{split}
\prod_{j=0}^{M-1} \operatorname{det}\left(H_w^{[N]}(\theta,\phi,\vartheta+j/L) - \mathscr E\right)  \le e^{(1-\delta) 4MN  \log\rho },
\end{split}
\]
where the left hand side is a Fourier polynomial of degree at most $4N^2 M.$ Setting $M=N^{1/2}$ we see that $\mathcal B_N(1/L):=\mathcal B_N^{N^{1/2}}(1/L)$ is a semi-algebraic set of degree at most $4N^{5/2}.$ We can then use \cite[Corollary $9.7$]{B} to see that for $N_1:=N^p $ with $p>0$ large enough
\[ \left\vert \{k=0,\ldots,N_1: \vartheta+\tfrac{k}{L} \in \mathcal  B_N(1/L)\} \right\vert<N_1^{(1-\delta(L))} \]
 for some small $\delta(L)>0$ and $N$ large enough.
Thus, $\vartheta+\tfrac{k}{L} \notin  \mathcal  B_N(1/L)$ is common. This implies by \eqref{eq:good_estm} that for every $\vartheta \in \mathbb T$ we can find $n \in [0,N_1)$ such that $G_N(\theta,\phi,\vartheta+\tfrac{n}{L},\mathscr E)$ is a good Green's function by \eqref{eq:good_green_function}.
\end{proof}

\subsection{Paving good Green's functions}
To finish the localization argument one wants to cover any large interval $\mathcal N'$ by smaller intervals $\{\mathcal N_n\}$, with $\mathcal N_n=[N]+j_n$, of length $\vert \mathcal N_n \vert=N$ as discussed in the previous subsection on which the Green's function exhibits good decay properties. More precisely, let $N'=N^C \gg N \gg 1$ with $N'$ a much larger number for $C>1$ than $N$, then, as we will see, any interval $\mathcal N' \supset [\tfrac{N'}{2},2N']$ of length of order $N'$, i.e. $N' \lesssim \vert \mathcal N' \vert \lesssim N'$ can be covered by a collection $\{\mathcal N_n\}$ of length $\vert\mathcal N_n \vert=N$ such that $G_w^{\mathcal N_n}( \mathscr E)$ exhibits exponential decay as in Proposition \ref{prop:good_decay}. For simplicity, we omit $\theta,\phi,\vartheta$ in this subsection.

The paving of Green's function follows by studying $\mathcal N = \mathcal N_1 \cup \mathcal N_2$, $\mathcal N_1 \cap \mathcal N_2 = \emptyset.$ Then by the resolvent identity  \cite[p.~60]{B}
\[  G_w^{\mathcal N}(\mathscr E)  = ( G_w^{\mathcal N_1}(\mathscr E) + G_w^{\mathcal N_2}(\mathscr E) )(\operatorname{id}-(H^{\mathcal N}_w-(H^{\mathcal N_1}_w+H^{\mathcal N_2}_w))G_w^{\mathcal N}(\mathscr E)).\]
Assuming $m \in \mathcal N_1$ and $n \in \mathcal N,$ this implies \cite[p.~60]{B}
\[\vert G_w^{\mathcal N}(m,n) \vert \le \vert G_w^{\mathcal N}(m,n) \vert \delta_{n \in \mathcal N_1}+\sum_{\substack{n' \in \mathcal N_1 \\ n'' \in \mathcal N_2, \vert n'-n'' \vert=1} } \vert G_w^{\mathcal N}(m,n') \vert \vert G_w^{\mathcal N}(n'',n) \vert. \]
This estimate shows that the concatenation of good Green's function is again a good Green's function. 

The final part of the argument for localization then consists of removing the energy dependence in the exceptional set $\mathcal  B_N(1/L).$ We know that $\mathcal  B_N(1/L)$ is a small set for every fixed energy, but as the sets could be disjoint for different energies, this could imply that for example every $\vartheta$ will eventually be in one of these sets for some energy. A key observation is now that it suffices to consider a finite set of energies determined by the union of spectra of finite-rank approximations of the Hamiltonian. This restriction is sufficient, since we already know that for any $\vartheta$ the Green's function will be (eventually) good by the last point in Proposition \ref{prop:good_decay}, we just don't know if good Green's functions can be paved together for general $\vartheta.$
Indeed, we shall study sets
\begin{equation}
\label{eq:SN}
S_N(1/L)=\bigcup_{\mathscr E \in \mathcal E(\vartheta)} \mathcal B_N(1/L,\mathscr E),
\end{equation}
with $\mathcal B_N$ as defined in the proof of Proposition \ref{prop:good_decay}. To see when the Green's function exhibits good decay, we must therefore study when $\vartheta+n/L \notin S_N(1/L)$ for $\sqrt{N'} \le n \le 2 N',$ where $N'$ is the large constant from the paving argument. This property is linked to the Green's function decay by Proposition \ref{prop:good_decay}.
The set of energies considered in \eqref{eq:SN} is just the union of all finite rank approximations
$$\mathcal E(\vartheta)= \bigcup_{1\le j \le N^p} \Spec(H^{[-j,j]}_w(\vartheta)).$$
Then, by a simple union bound $\vert \mathcal E(\vartheta) \vert \le \sum_{j=1}^{N^p} 4(2j+1) = 4N^p(N^p
+2).$

Notice that since the sets $\mathcal B_N$ are of exponentially small measure, this also implies that $\vert S_N(1/L) \vert =\mathcal O(e^{-N^{a'}})$ for any $0<a'<a.$

We then consider $\mathscr S_N:=\{(1/L,\vartheta): 1/L \in \operatorname{DC} \text{ and } \vartheta \in S_N(1/L)\}$ and aim to show that by discarding a suitable zero set of $1/L$, we may ensure that $\vartheta+n/L \notin S_N(1/L)$ for some $\sqrt{N'} \le n \le 2 N'$.

One then has by \cite[Chapter $10$]{B}, invoking the semi-algebraic sets, the estimate
\[ \Omega_N:=\{1/L \in \mathbb T: (1/L, n/L) \in \mathscr S_N \text{ for } n \sim N'\} \quad\Longrightarrow\quad \vert \Omega_N \vert \le 1/\sqrt{N'}.\]
Hence, since $N'=N^C$ with $C$ large, it follows that $\Omega:=\limsup_{N \to \infty} \Omega_N$ (in the measure-theoretic sense) is of zero measure. Hence, as long as a diophantine $1/L \notin \Omega$, then $1/L \notin \Omega_N$ for $N$ large enough. This yields the localization result of Theorem \ref{theo:Aloc1}.

\subsection{Arithmetic version of Anderson localization}
\label{sec:AVAL}
We shall now show how particle-hole symmetry from Lemma \ref{lem:particlehole} can be used to establish localization along the lines of Jitomirskaya's arithmetic argument for the AMO \cite{J99}. This argument provides a stronger localization result than the technique introduced by Klein which yields Theorem \ref{theo:Aloc1}. 
Arithmetic Anderson localization gives precise description of the localization frequencies and phases. To obtain arithmetic Anderson localization, we need to carefully analyze both frequency and phase resonances, see \cite{gy,gyzh1,jl1,jl2} and the references therein for more details. As Jitomirskaya's arithmetic localization argument heavily relies on the cosine-nature of the potential and requires non-resonant tunneling phases for A and B atoms of the potential, we will have to make more restrictive assumptions on the model. In particular, the argument does not seem to extend beyond the anti-chiral model.
We thus consider the Hamiltonian in the anti-chiral limit, but make three important simplifying assumptions that all seem to be critical (up to very simple generalizations) for the argument.
\begin{assumption}
\label{ass:assumption}
In this subsection, we consider the anti-chiral Hamiltonian with the following modifications:
\begin{itemize} 
\item We set $\phi=1/4$.
\item We assume $K(\theta):=1+e^{2\pi i \theta}\tau$, where we emphasize the occurrence of only the uni-directional shift operator.
\item The tunneling potential takes the form $U(x)=\cos(2\pi x).$ 
\end{itemize}
\end{assumption}
This leaves us with the Hamiltonian 
\begin{equation}
\label{eq:Hamiltonianmod}
 H_{w_0}(\theta,\phi=\tfrac{1}{4},\vartheta) = \begin{pmatrix} 0 &K(\theta)& w_0 U(\vartheta+\tfrac{\bullet-\phi}{L}) & 0 \\ 
K(\theta)^* &0 &0 & w_0U(\vartheta+\tfrac{\bullet+\phi}{L})   \\ 
w_0 U(\vartheta+\tfrac{\bullet-\phi}{L}) &0&0 &K(\theta)\\ 
0&w_0U(\vartheta+\tfrac{\bullet+\phi}{L})   &  K(\theta)^*&0   \end{pmatrix}. 
\end{equation}
Let us first comment on the applicability of Jitomirskaya's method for matrix-valued operators. In general, this method does not seem to apply well to matrix-valued operators. However, for the particular Hamiltonian \eqref{eq:Hamiltonianmod}, the characteristic polynomial of $H_{w_0}$ restricted to $N$ lattice sites will be, as we will show, a Fourier polynomial of degree $4N.$ Since the polynomial is also even, it suffices to study this polynomial at $2N$ distinct points. The definition of the Hamiltonian with $\phi =1/4$, implies that there are natural $2N$ values of the characteristic polynomial of shifts of the matrix at which we can interpolate the characteristic polynomial of the matrix. We say $\vartheta$ is Diophantine if 
$$
\inf_{j\in\mathbb{Z}}|2\vartheta-k/2L-j/2|\geq \frac{\kappa}{(1+|k|^\tau)},  \text{ for all } k\in\mathbb{Z}
$$
and some $\kappa>0,\tau>1.$
 \begin{theo}
 \label{theo:Jitomirskaya}
 Let $w_0$ be sufficiently large and $1/L \in \mathbb{T}$ diophantine. Then, under Assumption \ref{ass:assumption}, the Hamiltonian \eqref{eq:Hamiltonianmod} exhibits Anderson localization for Diophantine $\vartheta$.
 \end{theo}
 We shall now sketch the proof of Theorem \ref{theo:Jitomirskaya} emphasizing the main steps and differences compared with \cite{J99}.
 In order to flip lattices $2$ and $3$, we conjugate the Hamiltonian by $P = \operatorname{diag}(1,\sigma_1,1)$ such that $\mathscr H_{w_0}(\vartheta) := PH_{w_0}(\theta,\phi=\tfrac{1}{4},\vartheta)P$ becomes
 \[\mathscr H_{w_0}(\vartheta)  = \begin{pmatrix} 0  & w_0U(\vartheta+ \tfrac{n-\phi}{L})  & K(\theta)& 0 \\ 
 w_0U(\vartheta+\tfrac{n-\phi}{L}) &0 &0 &K(\theta)   \\ 
K(\theta)^* &0&0 & w_0U(\vartheta+\tfrac{n+\phi}{L}) \\ 
0&K(\theta)^* &w_0U(\vartheta+\tfrac{n+\phi}{L})   &  0   \end{pmatrix}. 
 \]
 
This together with definition \eqref{pro+-} implies that for $\mathcal N=[n_1,n_2]$
 \[\mathscr H^{-\mathcal N}_{w_0}(\vartheta) =\begin{pmatrix} \mathscr U_{n_1}^-(\vartheta) & T_1& 0 &0&0 & \cdots &\cdots & 0  \\
T_1^* &  \mathscr U_{n_1}^+(\vartheta) &T_2 & 0& \cdots &\cdots &\cdots&\vdots \\ 
0& T_2^*&  \mathscr U_{n_1+1}^-(\vartheta) &T_1& \ddots& \cdots& \cdots & \vdots   \\ 
\vdots & 0 & T_1^* &\mathscr U_{n_1+1}^+(\vartheta) &T_2 & 0 & 0 & \vdots \\ 
\vdots & \ddots & \ddots& T_2^*& \ddots &\ddots& 0 & \vdots \\ 
0 & \cdots & 0 &0& \ddots &  \ddots & T_2 & 0 \\
0 & \cdots & 0 &0& 0&T_2^*&  \mathscr U_{n_2}^-(\vartheta ) & T_1 \\
0 & \cdots & 0 &0& 0&0&  T_1^* & \mathscr U_{n_2}^+(\vartheta ) \\
 \end{pmatrix}    \]
 and
\[\mathscr H^{+\mathcal N}_{w_0}(\vartheta) =\begin{pmatrix} \mathscr U_{n_1}^+(\vartheta) & T_2& 0 &0&0 & \cdots &\cdots & 0  \\
T_2^* &  \mathscr U_{n_1+1}^-(\vartheta) &T_1 & 0& \cdots &\cdots &\cdots&\vdots \\ 
0& T_1^*&  \mathscr U_{n_1+1}^+(\vartheta) &T_2& \ddots& \cdots& \cdots & \vdots   \\ 
\vdots & 0 & T_2^* &\mathscr U_{n_1+1}^-(\vartheta) &T_1 & 0 & 0 & \vdots \\ 
\vdots & \ddots & \ddots& T_1^*& \ddots &\ddots& 0 & \vdots \\ 
0 & \cdots & 0 &0& \ddots &  \ddots & T_1 & 0 \\
0 & \cdots & 0 &0& 0&T_1^*&  \mathscr U_{n_2}^+(\vartheta ) & T_2 \\
0 & \cdots & 0 &0& 0&0&  T_2^* & \mathscr U_{n_2+1}^-(\vartheta ) \\
 \end{pmatrix}    \]
with $\mathscr U_n^{\pm}(\vartheta) =w_0 U(\vartheta + \frac{n\mp \phi}{L})\sigma_1$, $T_1 = I_{2}$ and $T_2 = e^{2\pi i \theta} T_1.$
Then we define the block-diagonal Pauli matrix 
\[V = \operatorname{diag}(\sigma_3,\ldots,\sigma_3) \]
which shows that 
\begin{align}\label{we1}
\mathscr H^{\pm, \mathcal N}_{-w_0}(\vartheta) =V\mathscr H^{\pm,\mathcal N}_{w_0}(\vartheta) V.
\end{align}
In addition, we have that for $\sigma_0= I_{2}$ with
\[W = \operatorname{diag}(\sigma_0,e^{2\pi i\theta} \sigma_0,\sigma_0, e^{2\pi i\theta}\sigma_0,\ldots,\sigma_0,e^{2\pi i\theta}\sigma_0)\]
that
\begin{align}\label{we2}
W^*\mathscr H^{+,\mathcal N}_{w_0}(\vartheta)W =\mathscr H^{-,\mathcal N}_{w_0}(\vartheta-1/(2L)).
\end{align}
We recall the definition of $(\gamma,k)$-regularity which we shall apply to our operator $\mathscr H_{w_0}(\vartheta).$ Note that in this case
\begin{gather*} G^{\pm, \mathcal N}_{w_0}(\vartheta,\mathscr E)_{\alpha,\alpha'}=( H^{\pm,\mathcal N}_{w_0}(\vartheta)-\mathscr E)^{-1}_{\alpha,\alpha'},
\\
 \mu^{\pm, \mathcal N}_{(\alpha,\alpha')}(\vartheta,w_0,\mathscr E):=\operatorname{det}(((H_{w_0}^{\pm, \mathcal N}(\vartheta)-\mathscr E)_{\beta,\beta'})_{\beta\ne\alpha,\beta'\ne\alpha'}).
 \end{gather*}
\begin{defi}[$(\gamma,k)$-regularity] Let $\mathscr E,\gamma \in \mathbb R$ and $k \ge 1$. We call a number $n \in \mathbb Z$ $(\gamma,k)$-regular if there is $ \mathcal N=[n_1,n_2] \subset \ZZ$, with $n \in \mathcal N$ such that
\begin{itemize}
\item $n_2 = n_1 + k-1,$ $I = \{4n_1,4n_2+3\},$
\item $\alpha =4n \in [4n_1,4n_2+3]$, $d(I,4n) > \frac{4k}{5},$
\item  either $\vert G^{-, \mathcal N}_{w_0}(\vartheta,\mathscr E)_{\alpha,\alpha'}\vert < e^{-\gamma \vert \alpha-\alpha' \vert}$ or $\vert G^{+, \mathcal N}_{w_0}(\vartheta,\mathscr E)_{\alpha,\alpha'}\vert < e^{-\gamma \vert \alpha-\alpha' \vert}$ where $\alpha' \in I.$
\end{itemize} 
\end{defi}
If $(\gamma,k)$ is not \emph{regular}, we call it \emph{singular}. In particular, for $k$ sufficiently large and $\gamma$ fixed, it is clear that any point of $y \in \mathbb Z$ such that $u(y) \neq 0$, for $u$ a generalized eigenfunction, is $(\gamma,k)$-singular.

Observe that for $\mathcal N=[0,N-1]$ the characteristic polynomial 
$$p^{\pm, \mathcal N}(\vartheta)=\det(\mathscr H_{w_0}^{\pm, \mathcal N}(\vartheta)-\mathscr E)
$$
has the property that  $p^{\pm, \mathcal N+1}(\vartheta) = p^{\pm, \mathcal N}(\vartheta + \tfrac{1}{L}).$ In addition, in the case of $\phi=1/4$, by \eqref{we2}, we have $p^{+,\mathcal N}(\vartheta) = p^{-,\mathcal N}(\vartheta - \tfrac{1}{2L})$. 
Hence by \eqref{we1}, $\vartheta \mapsto p^{-,\mathcal N}(\vartheta - \tfrac{N-1}{2L})$ is an even function that satisfies $p^{-,\mathcal N}(\vartheta)= p^{-,\mathcal N}(\vartheta+\frac{1}{2})$, which is an additional symmetry that does not exist in the case of the AMO, and therefore by the orthogonality of the Fourier basis 
\[ p^{-,\mathcal N}(\vartheta -  \tfrac{N-1}{2L}) = \sum_{j \in [0,2N]} \tilde{b}_j(L) \cos\left(4\pi j \vartheta \right) \text{ for }\tilde{b}_j(L) \in \RR,\]
such that for some new $b_j \in \RR$
\[ q\big(\cos\left(2\pi (\vartheta+  \tfrac{N-1}{2L})\right)^2\big):=p^{-,\mathcal N}(\vartheta)= \sum_{j \in [0,2N]} b_j \cos^{2j}\left(2\pi (\vartheta+  \tfrac{N-1}{2L}) \right).\]
We observe that $q$ is a polynomial of degree $2N.$

\begin{lemm}
\label{lemm:techniqual}
Suppose $n \in \mathbb Z$ is $(\gamma,k)$-singular, then for any $j$ with 
\[ n-\lceil \tfrac{3}4 k \rceil \le j \le n-\lfloor \tfrac{3}4k\rfloor+ \tfrac{k+1}{2},\]
we have that for $\mathcal N = [j,k+j]$, 
\[ \vert \operatorname{det}(\mathscr H_{w_0}^{-, \mathcal N}(\vartheta)-\mathscr E) \vert \le e^{4 k( \vert \log w_0\vert +\frac{ \gamma - 4\vert \log w_0\vert}{5} + \mathcal O(1))}.
\]
\end{lemm}
\begin{proof}
Let $n_1:=j$ and $n_2:=k+j$. By the definition of singularity, Cramer's rule \eqref{eq:Cramer}, and \eqref{eq:Klein_bound}, we have for $n_i \in \{n_1,n_2\}$ and $n$ as in the definition of singularity
\[
\begin{aligned}
\vert \operatorname{det}(\mathscr H_{w_0}^{ \pm, \mathcal N}(\vartheta)-\mathscr E) \vert &\le \exp{(\gamma \vert \alpha-\alpha_i \vert)}\vert \mu_{(\alpha,\alpha_i)}^{\pm,\mathcal N}(w_0,\vartheta,\mathscr E)\vert \\
&\le \exp{(\gamma \vert \alpha-\alpha_i \vert)} \exp{\left(4k\vert\log w_0\vert\left(1- \tfrac{\vert n-n_i\vert}{4k}\right) + \mathcal O(4k)\right)}\\
&\le \exp{(4k\vert\log w_0\vert)} \exp{\big((\gamma-4\vert\log w_0\vert) \vert \alpha-\alpha_i \vert + \mathcal O(4k)\big)}\\
&\le \exp{\left(4 k\left( \vert \log w_0\vert +\frac{ \gamma - 4\vert \log w_0\vert}{5} + \mathcal O(1)\right)\right)},
\end{aligned}
\]
where we used that $\vert \alpha-\alpha_i \vert \ge 4k/5$ and $\alpha_i=4n_i$ is the index that corresponds to $n_i.$
\end{proof}

Let $n_1$ and $n_2$ be both $(\gamma,k=N)$-singular, $d :=n_2-n_1>\frac{k+1}{2}$, and $x_i = n_i - \lfloor{\tfrac{3k}{4} \rfloor}.$ 
We now set, for $ \phi=1/4$
\begin{equation}
\begin{split}
\label{eq:vartheta}
\vartheta_j:= \begin{cases}
\vartheta + \tfrac{n_1 + \frac{k-1}{2} +\frac{j}{2}}{L} , &j=0,\ldots,2\lceil \frac{k+1}{2} \rceil -1 \\
\vartheta + \tfrac{n_2+ \frac{k-1}{2}+\frac{j}{2}-\lceil \frac{k+1}{2} \rceil}{L}, &j =2 \lceil \frac{k+1}{2}\rceil,\ldots,2k.
\end{cases}
\end{split}
\end{equation}
By the assumption that $d > \frac{k+1}{2}$ all $\vartheta_j$ are distinct. 
Lagrange interpolation yields then, since $p^{\mathcal N^-}$ is even with respect to $ \vartheta + \tfrac{N-1}{2L}$,
\begin{equation}
\label{eq:interp}
 q(z^2) = \sum_{j \in [0,2N]} q(\cos(2\pi \vartheta_j)^2) \frac{\prod_{l \neq j}(z^2-\cos^2(2\pi \vartheta_l))}{\prod_{l \neq j} (\cos^2(2\pi \vartheta_j) - \cos^2(2\pi \vartheta_l))}.
 \end{equation}

Finally, we have by \cite[Lemma $7$]{J99}, \cite[Lemma 5.8]{aj} for $d < k^{\alpha}$ with $\alpha<2$ that for {Diophantine $\vartheta$} and any $\varepsilon>0$ there is $K>0$ such that for $k>K$ and all $z \in [-1,1]$
\[ \left\lvert \frac{\prod_{l \neq j} (z-\cos(2\pi \vartheta_l)) }{\prod_{l \neq j} (\cos(2\pi \vartheta_j)\pm \cos(2\pi \vartheta_l))}\right\rvert \le e^{k \varepsilon}.\]
Combining this estimate with Lemma \ref{lemm:techniqual}, which applies to the above choice of \eqref{eq:vartheta}, and the interpolation formula \eqref{eq:interp}, we thus conclude that for all $\vartheta \in [-1,1]$
\begin{equation}
\label{eq:ub}
\vert p^{\pm,\mathcal N}(\vartheta) \vert \le (2k+1)e^{4 k( \vert \log w_0\vert +\frac{ \gamma - \vert 4\log w_0\vert}{5} + \mathcal O(1))}.
\end{equation}
In addition, we have by \cite[Proposition $3$]{K17} and \eqref{eq:as_Lyap} the existence of some $\vartheta_0 \in \mathbb T$  such that for $k$ large enough
\begin{equation}
\label{eq:lb}  \vert p^{\pm,\mathcal N}(\vartheta_0) \vert \ge  e^{4k (\vert \log w_0\vert-O(1))}.
\end{equation}
Having both \eqref{eq:ub} and \eqref{eq:lb} leads to a contradiction to our assumption $d<k^{\alpha}$ for two singular clusters, once $w_0$ is large enough with $\gamma=\log w_0$. Hence, singular points are far apart. Fixing an energy $\mathscr E \in \mathbb R$, and $u_{\mathscr E}$ a generalized eigenfunction to $\mathscr E$ of the operator with $u_{\mathscr E}(0) \neq 0.$ The last condition can be assumed without loss of generality, as $u_E$ may not vanish at $u_{\mathscr E}(-1)$ and $u_{\mathscr E}(0)$ at the same time. As mentioned before, since $u_{\mathscr E}(0)\neq 0$, $0$ has to be $(\log w_0,k)$ singular for $k$ sufficiently large. Hence, repulsion of singular clusters shows that any $n$ sufficiently large will be  $(\log w_0,k)$ regular for some suitable $k.$ Thus, we obtain an interval $\mathcal N =[n_1,n_2]$ of length $\vert \mathcal N \vert=k$ with $n \in \mathcal N$ such that $$ \frac{1}{5} (\vert n \vert-1) \le \vert n_i-n \vert \le \frac{4}{5}(\vert n \vert-1)$$ with $n_i \in \{n_1,n_2\}$ and the decay bound
$$ \Vert G^{-, \mathcal N}_{w_0}(\vartheta,\mathscr E)_{\alpha,\alpha_i} \Vert \le e^{-\frac{\log w_0}{2} \vert \alpha- \alpha_i \vert},\ \ \text{or}\ \ \Vert G^{+, \mathcal N}_{w_0}(\vartheta,\mathscr E)_{\alpha,\alpha_i} \Vert \le e^{-\frac{\log w_0}{2} \vert \alpha- \alpha_i \vert}.$$
Combining this with \eqref{eq:GFestm}, we find for any generalized eigenfunction $u_{\mathscr E}$
\[ \vert u_{\mathscr E}(n) \vert \le  2 C \langle n \rangle e^{-\log w_0 (\vert n \vert-1)/100}.\]
This implies exponential localization and finishes the sketch of proof of Theorem \ref{theo:Jitomirskaya}.

\section{Weak coupling and AC spectrum}
\label{sec:AC}
We now return to the usual Hamiltonian \eqref{eq:Hamiltonian} and study the regime where the coupling of the honeycomb lattices is weak and see that the AC spectrum that is present in case of non-interacting sheets persists. We saw in Theorem \ref{theo:Aloc1} in the previous section that in the strong coupling regime, the Hamiltonian exhibits Anderson localization (point spectrum) at almost every diophantine moir\'e lengths $1/L$. In this section, we show for {\it fixed} diophantine moir\'e lengths $1/L$, $H_w(\theta,\phi,\vartheta)$ has some AC spectrum if the coupling is weak enough. Our main theorem is then as follows.

\begin{theo}
\label{theo:AC}
For $\frac{1}{L}\in \operatorname{DC}_t$ and small enough coupling $\vert w \vert \le c$, for some constant $c(L)>0,$ the AC spectrum of the Hamiltonian $H_w(\theta,\phi,\vartheta)$ is non-empty.
\end{theo}
This result follows immediately from Theorem \ref{ac} which will be proved in this section.

Recall that  the famous Schnol's theorem says the spectrum of $H_w(\theta,\phi,\vartheta)$ is given by the closure of the set of generalized eigenvalues of $H_w(\theta,\phi,\vartheta)$. Actually, one can also characterize the AC spectrum based on more concrete descriptions of growth of the generalized eigenfunctions. The theory was first built for one dimensional discrete Schr\"odinger operators. Let $H$ be a discrete Schr\"odinger operator on $\ell^2(\mathbb{Z})$:
\[
(Hu)(n)=u(n-1)+u(n+1)+V_nu(n),\ \ n\in\mathbb{Z},
\]
where $\{V_n\}_{n\in \mathbb{Z}}$ is a sequence of real numbers (the potential). A non-trivial solution $u$ of $Hu=\mathscr Eu$ is called subordinate at $\infty$ if
$$
\lim\limits_{L\rightarrow\infty}\frac{\|u\|_L}{\|v\|_L}=0
$$
for any linearly independent solution $v$ of $Hu=Eu$, where 
$$
\|u\|_L=\left[\sum\limits_{n=1}^{[L]}|u(n)|^2+(L-[L])|u([L]+1)|^2\right]^{\frac{1}{2}},
$$
here $[L]$ denotes the integer part of $L$. The absolutely continuous spectrum of $H$, denoted by $\Spec_{ac}(H)$, has the following characterization,
\begin{align*}
\Spec_{ac}(H)=\overline{\left\{\mathscr E\in\mathbb{R}: \text{at $\infty$ or $-\infty$, $Hu=\mathscr Eu$ has no  subordinate solution}\right\}}^{ess},
\end{align*}
 known as the subordinate theory \cite{GP,JL99,LS99}. For our purpose, we will use the subordinate theory \cite{OC2021} for the following matrix-valued Jacobi operators,
\[
(J{u})(n)=D_{n-1}{u}(n-1)+D_n {u}(n+1)+V_n{u}(n),\ \ n\in\mathbb{Z},
\]
where $(D_n)_n, (V_n)_n$ are bilateral sequences of $m\times m$ self-adjoint matrices. 

We define the Dirichlet and Neumann solutions as the solutions to $J{u}=\mathscr E{u}$ that satisfy, respectively, the initial conditions
$$
\begin{cases}
\phi_0=0_m,\\
\phi_{-1}=I_m,
\end{cases}\ \ \ \ 
\begin{cases}
\psi_0=I_m,\\
\psi_{-1}=0_m.
\end{cases}
$$
where $I_m$ and $0_m$ are the m-dimensional identity and zero matrices, respectively.
\begin{theo}[\cite{OC2021}]\label{criteria}
Let, for each $r\in\left\{1,2\cdots, m\right\}$,
$$
\mathcal{S}_r=\left\{\mathscr E\in\mathbb{R}:\liminf\limits_{L\rightarrow \infty}\frac{1}{L}\sum\limits_{n=1}^L \sigma^2_{m-r+1}(\phi_n(\mathscr E))+\sigma^2_{m-r+1}(\psi_n(\mathscr E))<\infty\right\},
$$
where $\sigma_k(T)$ stands for the $k$-th singular value of $T$.
Then the set $\overline{\mathcal{S}_{r+1}\backslash\mathcal{S}_r}^{ess}$ corresponds to the absolutely continuous component of multiplicity $r$ of any self-adjoint extension of the operator $J^+$ which is $J$ restricted to ${\operatorname{Dom}}(J^+_{max}):=\{{u}\in \ell^2(\mathbb{N};\mathbb{C}^m)| J^+{u}\in \ell^2(\mathbb{N};\mathbb{C}^m)\}$ (satisfying any admissible boundary condition at $n=0$).
\end{theo}

Thus, to characterize AC spectrum for matrix-valued Schr\"odinger operators, one needs study the singular values of the transfer matrix. In this section, we are mainly interested in the quasi-periodic case. We consider a rather general quasiperiodic model and the Hamiltonian in \eqref{eq:Hamiltonian} is one of the typical examples. For our purpose, consider the following multi-frequency matrix-valued Schr\"odinger operators,
\begin{equation}\label{eq:J}
(J_{\vartheta}{u})_n=C{u}_{n+1}+C{u}_{n-1}+ V(\vartheta+n\alpha){u}_n,
\end{equation}
acting on $\ell^2(\mathbb{Z};\mathbb{C}^m)$ where $\alpha\in(\mathbb{R}\backslash\mathbb{Q})^d$, $C$ is a $m\times m$ invertible self-adjoint matrix and $V$ is an analytic self-adjoint matrix which is 1-periodic in each variable. In particular, {$\alpha=1/L$, $C=t(\theta)$ and $V(x)=V_{w}^\phi(x)$} in the case of the Hamiltonian \eqref{eq:Hamiltonian}. 

Our approach is the so-called {\it reducibility} method, which was initially developed by Dinaburg-Sinai \cite{ds}, Eliasson \cite{e}, further developed by Hou-You \cite{hy}, Avila-Jitomirskaya \cite{aj} and Avila \cite{avila1,avila2}. Our result is based on the reducibility results in \cite{eliasson,chavaudret,gyzh,He-You} for higher dimensional quasi-periodic cocycles.

\begin{rem}
Subordinate theory can be used to give an alternative proof of Proposition \ref{prop:rational}. In fact, Theorem \ref{criteria} implies that only the end points of each spectral interval can support singular spectrum, while there are only finitely many such points, so they can only support point spectrum. Point spectrum can then be ruled out using our extension of Gordon's argument \cite{G76}, see the proof of Proposition \ref{prop:Liouville}.
\end{rem}

\subsection{Preliminaries}
Given $A \in C^\omega(\mathbb{T}^d, \operatorname{GL}(2m,\mathbb{C}))$ and rational independent $\alpha\in\mathbb{R}^d$, we define the {\it quasi-periodic $\operatorname{GL}(2m,\mathbb{C})$ cocycle} $(\alpha,A)$:
$$
(\alpha,A)\colon \left\{
\begin{array}{rcl}
\mathbb{T}^d \times \mathbb{C}^{2m} &\to& \mathbb{T}^d \times \mathbb{C}^{2m}\\[1mm]
(x,v) &\mapsto& (x+\alpha,A(x)\cdot v).
\end{array}
\right.  
$$
We denote by $L_1(\alpha, A)\geq L_2(\alpha,A)\geq\ldots\geq L_{2m}(\alpha,A)$ the Lyapunov exponents of $(\alpha,A)$ repeated according to their multiplicities, i.e.,
$$
L_k(\alpha,A)=\lim\limits_{n\rightarrow\infty}\frac{1}{n}\int_{\mathbb{T}}\log(\sigma_k(A_n(x)))dx.
$$

$(\alpha,A)$ is said to be reducible if there exist $B\in C^\omega(\mathbb{T}^d,\operatorname{GL}(2m,\mathbb{C}))$, $\widetilde{A}\in \operatorname{GL}(2m,\mathbb{C})$ such that
$$
B(x+\alpha)A(x)B^{-1}(x)=\widetilde{A}.
$$

The following are two general facts on the Lyapunov exponents, which were proved in \cite{gyzh} (see Proposition 2.2 and Proposition 2.3).
\begin{prop}\label{invariance}
Assume $(\alpha,A)\in \mathbb{T}^d\times C^0(\mathbb{T}^d,\operatorname{GL}(2m,\mathbb{C}))$, $B\in C^0(\mathbb{T}^d,\operatorname{GL}(2m,\mathbb{C})),$ and $\widetilde{A}(x)=B(x+\alpha)A(x)B^{-1}(x)$, then 
$$
L_i(\alpha,\widetilde{A})=L_i(\alpha,A),\ \ 1\leq i\leq 2m.
$$
\end{prop}

\begin{prop}\label{leconstant}
If we denote the eigenvalues of $A\in \operatorname{GL}(2m,\mathbb{C})$ by $\{e^{-2\pi i \rho_j}\}_{j=1}^{2m}$, then
$$
\{L_j(\alpha,A)\}_{j=1}^{2m}=\{2\pi\Im \rho_j\}_{j=1}^{2m}.
$$
\end{prop}

Now we consider the eigen-equation $J_{\vartheta}{u}=\mathscr E{u}$ with $J_{\vartheta}$ as in \eqref{eq:J}. To obtain a first order system and the corresponding linear skew product we use the fact that $C$ in \eqref{eq:J} is invertible  and write
$$
\begin{pmatrix}
{u}_{k+1}\\
{u}_k
\end{pmatrix}
=\begin{pmatrix}
C^{-1}(\mathscr EI_m-V(\vartheta+k\alpha))& -I_m\\
I_m&0_m
\end{pmatrix}
\begin{pmatrix}
{u}_k\\
{u}_{k-1}
\end{pmatrix}.
$$
Denote
\begin{equation*}
L^{V}_{\mathscr E}(\vartheta)=\begin{pmatrix}
C^{-1}(\mathscr EI_m-V(\vartheta))& -I_m\\
I_m&0_m
\end{pmatrix}.
\end{equation*}
Note that  $(\alpha, L_{\mathscr E}^{V})$ is a symplectic cocycle. As a corollary,  the Lyapunov exponents of  $(\alpha,L_{\mathscr E}^{V})$ come in pairs $\pm L_i(\alpha,L_{\mathscr E}^{V})$ ($1\leq i\leq m$).

Let  ${\operatorname {Sp}_{2m}}(\mathbb{C})$ denote the set of $2m\times 2m$ complex symplectic matrices. Given any  $A\in C^0(\mathbb{T}^d,{\operatorname {Sp}_{2m}}(\mathbb{C}))$,  we say the  cocycle $(\alpha, A)$ is {\it uniformly hyperbolic} if for every $x \in \mathbb{T}^d$, there exists a continuous splitting $\mathbb{C}^{2m}=E^s(x)\oplus E^u(x)$ such that for some constants $C>0,c>0$, and for every $n\geqslant 0$,
$$
\begin{aligned}
\lvert A_n(x)v\rvert \leqslant Ce^{-cn}\lvert v\rvert, \quad & v\in E^s(x),\\
\lvert A_n(x)^{-1}v\rvert \leqslant Ce^{-cn}\lvert v\rvert,  \quad & v\in E^u(x+n\alpha).
\end{aligned}
$$
This splitting is invariant by the dynamics, which means that for every $x \in \mathbb{T}^d$, $A(x)E^{\bullet}(x)=E^{\bullet}(x+\alpha)$, for $\bullet=s,u$.   The set of uniformly hyperbolic cocycles is open in the $C^0$-topology.

 Let $\Sigma$ be the spectrum of $J_{\vartheta}$.  $\Sigma$ is closely related to the dynamical behavior of  the symplectic cocycle  $(\alpha,L_{\mathscr E}^{V})$.  $\mathscr E\notin \Sigma$ if and only if $(\alpha,L_{\mathscr E}^{V})$ is uniformly hyperbolic.

\subsection{Existence of AC spectrum}
Let $F$ be a bounded analytic (possibly matrix-valued) function defined on $\{\theta\in\TT:|\Im \theta|< h \}$, and let
$|F|_h= \sup_{|\Im\theta|< h} |F(\theta)|$. When $F$ is matrix-valued, we understand the absolute value on the right-hand side as the operator norm. $C^\omega_{h}(\mathbb{T};\bullet)$ denotes the
set of all these $\bullet$-valued functions ($\bullet$ will usually denote $\mathbb{C}$, $\operatorname{GL}(m,\mathbb{C})$ or
$\operatorname{SL}(2,\mathbb{R})$). Denote by $C^{\omega}(\mathbb{T};\bullet)$  the union $\bigcup_{h>0}C_h^{\omega}(\mathbb{T};\bullet)$.

In this subsection, we will prove the following theorem.
\begin{theo}\label{ac}
Let $\alpha\in {\operatorname{DC}}^d_t$ and $V$ be an analytic self-adjoint $m\times m$ matrix. Then there is $\varepsilon_0(m,\alpha,C,V)$ such that if $|V|_h<\varepsilon_0$, $\theta\in\mathbb{R}$ is the phase and $V\in C_h^\omega(\mathbb{T};\mathbb{R})$, then the ac part of the spectrum of $J_{\vartheta}$ is non-empty for any $\vartheta\in\mathbb{T}^d$.
\end{theo}

The proof is based on a positive measure reducibility theorem for higher dimensional quasi-periodic cocycles and subordinate theory for matrix-valued Schr\"odinger operators.

\begin{theo}\label{red}
Let $\alpha\in {\operatorname{DC}}^d_t$ and $V$ be an analytic self-adjoint $m\times m$ matrix. There is $\varepsilon_0(m,\alpha,C,V)$  and $\mathcal{E}\subset \Sigma$ such that if $|V|_h<\varepsilon_0$, then for any $\mathscr E\in \mathcal{E}$, $\left(\alpha, L_{\mathscr E}^{V}\right)$ is reducible. Moreover, $\left|\Sigma\backslash\mathcal{E}\right|\rightarrow 0$ as $|V|_h\rightarrow 0$.
\end{theo}

\noindent{\it Proof of Theorem \ref{ac}.}
Under the assumption, by Theorem \ref{red}, for any $\mathscr E\in \mathcal{E}$, there is $B\in C^\omega(\mathbb{T}^d,\operatorname{GL}(2m,\mathbb{C}))$, $A_\mathscr E\in \operatorname{GL}(2m,\mathbb{C})$ such that
\begin{equation}\label{redu}
B(\vartheta+\alpha)L_{\mathscr E}^{V}(\vartheta)B^{-1}(\vartheta)=A_\mathscr E.
\end{equation}
By  Proposition \ref{invariance} and Proposition \ref{leconstant},
$$
\{L_j(\alpha, L_\mathscr E^{V})\}_{j=1}^m=\{\log|\lambda_j(\mathscr E)|\}_{j=1}^{m},
$$
where $\lambda_1(\mathscr E), \cdots, \lambda_m(\mathscr E)$ are the eigenvalues of  $A_\mathscr E$ outside the unit circle, counting multiplicity.

Since $\mathcal{E}\subset \Sigma$, we claim $|\lambda_m(\mathscr E)|=1$ for all $\mathscr E\in \mathcal{E}$. Otherwise, by the definition, $(\alpha,L_\mathscr E^{V})$ is uniformly hyperbolic which contradicts that $\mathscr E\in \Sigma$. We let
$$
 2\omega(\mathscr E)=\text{the number of unit eigenvalues of $A_\mathscr E$}.
$$
Then by the above argument and the complex symplectic structure we have $\omega(\mathscr E)$ is always an integer and $\omega(\mathscr E)\geq 1$.

To characterize the AC spectrum, we need to apply Theorem \ref{criteria}. We define the Dirichlet and Neumann solutions as the solutions to $J_{\vartheta}{u}=\mathscr E{u}$ that satisfy, respectively, the initial conditions
$$
\begin{cases}
\phi_0(\vartheta,\mathscr E)=0_m,\\
\phi_{-1}(\vartheta,\mathscr E)=I_m,
\end{cases}\ \ \ \ 
\begin{cases}
\psi_0(\vartheta,\mathscr E)=I_m,\\
\psi_{-1}(\vartheta,\mathscr E)=0_m.
\end{cases}
$$
Note that
$$
\begin{pmatrix}
\phi_{n}(\vartheta,\mathscr E)\\
\phi_{n-1}(\vartheta,\mathscr E)
\end{pmatrix}
=(L_{\mathscr E}^{V})_n(\vartheta)
\begin{pmatrix}
I_m\\
0_m
\end{pmatrix},
\quad
\begin{pmatrix}
\psi_{n}(\vartheta,\mathscr E)\\
\psi_{n-1}(\vartheta,\mathscr E)
\end{pmatrix}
=(L_{\mathscr E}^{V})_n(\vartheta)
\begin{pmatrix}
0_m\\
I_m
\end{pmatrix}.
$$
On the other hand, by \eqref{redu}
$$
(L_{\mathscr E}^{V})_n(\vartheta)=B^{-1}(\vartheta+n\alpha)A_E^nB(\vartheta).
$$
Thus there is $\bar{C}$ depending $C,V$ such that
$$
 \sigma^2_{m-r+1}(\phi_n(\vartheta, \mathscr E))+\sigma^2_{m-r+1}(\psi_n(\vartheta,\mathscr E))\leq \bar{C}, \ \ 1\leq r\leq \omega(\mathscr E).
$$
which implies $\mathcal{E}\subset\mathcal{S}_1$. By Theorem \ref{criteria}, we have
$$
\overline{\mathcal{E}}^{ess}\subset \Spec_{ac}(J^+_{\vartheta})
$$
where $J^+_{\vartheta}$ is the restriction of $J_{\vartheta}$ on $\ell^2(\mathbb{Z}^+)$ with Dirichlet boundary condition. Actually, one can prove in the exact same way that 
$$
\overline{\mathcal{E}}^{ess}\subset \Spec_{ac}(J^-_{\vartheta})
$$
where $J^-_{\vartheta}$ is the restriction of $J_{\vartheta}$ on $\ell^2(\mathbb{Z}^-)$ with Dirichlet boundary condition.  Finally, we need the following proposition in \cite{OC2021}.
\begin{prop}
We have $\Spec_{ac}(J_{\vartheta})=\Spec_{ac}(J_{\vartheta}^+\oplus J_{\vartheta}^-)$.
\end{prop}
Hence
$$
\overline{\mathcal{E}}^{ess}\subset \Spec_{ac}(J_{\vartheta}).
$$
Theorem \ref{ac} follows from  the fact that $|\mathcal{E}|>0$.
\hfill$\square$

\noindent {\it Proof of Theorem \ref{red}.} First, we introduce a positive measure reducibility theorem proved in \cite{He-You}. 
We consider $(\alpha, A(\mathscr E)+F(\mathscr E,\cdot))\in C^\omega(\mathbb{T},\operatorname{GL}(2m,\mathbb{C}))$, $A(\mathscr E)+F(\mathscr E,\cdot)$ also depends analytically on $\mathscr E\in I$ where $I$ is an interval. Let $\{\lambda_i(\mathscr E)\}_{i=1}^{2m}$ be the eigenvalues of $A(\mathscr E)$. For any $u\in\mathbb{R}$, we denote
$$
h(\mathscr E,u)=\prod\limits_{i\neq j}\left(\lambda_i(\mathscr E)-\lambda_j(\mathscr E)-iu\right).
$$
We say $A(\mathscr E)$ is non-degenerate if there is $p\in\mathbb{N}$ such that for all $u\in\mathbb{R}$, we have
\begin{align}\label{nonde}
\max_{1\leq i\leq p}\left|\frac{\partial^i h(\mathscr E,u)}{\partial \mathscr E^i}\right|\geq c>0.
\end{align}
\begin{theo}[\cite{He-You}]\label{pred}
Assume $\alpha\in \operatorname{DC}_t^d$ and $I$ is a parameter interval, $A\in C^\omega(I,\operatorname{GL}(2m,\mathbb{R}))$ satisfies the non-degeneracy condition \eqref{nonde}, {$F\in C^\omega(\mathbb T\times I,\operatorname{GL}(2m,\mathbb{R}))$} and there is  $M>0$ such that $|A|_s<M$. Then there exists $\varepsilon_0$ such that if $|F|_{s,\delta}<\varepsilon_0$\footnote{Here $|F|_{s,\delta}=\sup_{|\Im \vartheta|<s,|\Im \mathscr E|<\delta} |F(\mathscr E,\vartheta)|$.}, the measure of the set of parameter $I$ for which $(\alpha, A(\mathscr E)+F(\mathscr E,\cdot))$ is not reducible is no larger than $CL(10\varepsilon_1)^c$, where $C,c$ are some positive constants, $L$ is the length of the parameter interval $I$.
\end{theo}
Actually, Theorem \ref{red} is a special case of Theorem \ref{pred}. To see this, we denote the characteristic polynomial of $L^{0}_{\mathscr E}$ by $p(\mathscr E,z)=\det\left({L^{0}_{\mathscr E}-zI_{2m}}\right)$. Let $\{z_i(\mathscr E)\}_{i=1}^{2m}$ be the zeros of $p(\mathscr E,z)$.  For any $u\in\mathbb{R}$, we denote
$$
g(\mathscr E,u)=\prod\limits_{i\neq j}\left(z_i(\mathscr E)-z_j(\mathscr E)-iu\right).
$$
It is easy to check that
\begin{align*}
g(\mathscr E,u)&=\det{\left[iu I_{16m^2}-\left(I_{4m}\otimes L_\mathscr E^0-(L_\mathscr E^0)^\intercal\otimes I_{4m}\right)\right]}/(iu)^{4m}\\
&=\mathscr E^{2m}+g'(\mathscr E,u)
\end{align*}
where $g'(\mathscr E,u)$ is a polynomial of $\mathscr E$ with degree $\leq 2m-1$. Thus $g(\mathscr E,u)$ satisfies the following non-degeneracy condition for all $\mathscr E\in\mathbb{R}$
$$
\max_{1\leq i\leq {2m}}\left|\frac{\partial^i g(\mathscr E,u)}{\partial \mathscr E^i}\right|\geq 1.
$$
Moreover, $|L_\mathscr E^0|\leq M$ for all $\mathscr E\in \Sigma$. Thus all conditions in Theorem \ref{pred} are satisfied. This completes the proof of Theorem \ref{red}.\hfill$\square$

\section{Cantor spectrum for anti-chiral Hamiltonian}
\label{sec:Cantor}

The key property of the anti-chiral Hamiltonian that allows us to establish Cantor spectrum is the operator-valued diagonalizability of the matrix-valued discrete operator. In this section, we show that we can block-diagonalize the Hamiltonian into four Schr\"odinger operators for $\theta \in \ZZ/2.$ 

To this end, we define the matrix
\[\mathcal U=\frac{1}{2}\begin{pmatrix} 1 &-1& -1&1 \\-1 & 1 & -1 & 1 \\ -1 & -1 & 1 &1 \\ 1 & 1 & 1 & 1
\end{pmatrix}.\]
This implies that 
\[\mathcal U^*H_{(w_0,0)}(\theta,0) \mathcal U = 1 + e^{2\pi i \theta} (\tau+\tau^*)  \operatorname{diag}(-I_{2},I_{2})-w_0 U(\vartheta+\tfrac{\bullet}{L})\operatorname{diag}(\sigma_3,\sigma_3). \]
By flipping matrix entries $(1,1),(3,3)$ and $(2,2),(4,4)$ of each individual block appropriately, we see that the Hamiltonian is equivalent to the following operator on $\ell^2(\ZZ;\CC^4)$:
\[ w_0 \left( \frac{\tau^* + \tau-1}{3} \right)\operatorname{diag}(\sigma_3,\sigma_3) +\left(2 \cos\left(2\pi( \frac{\bullet}{L} +\vartheta)\right) \pm 1 \right)\operatorname{diag}(-I_{2},I_{2}),\]
where the choice of the $\pm$ sign depends on whether we are studying $\theta=0$ or $\theta=1/2.$ It is easy to see that this operator is a direct sum of (shifted) almost Mathieu operators (AMO).

\begin{figure}
\includegraphics[width=7.5cm,height=8cm, trim=0.5cm 6cm 0.5cm 6cm]{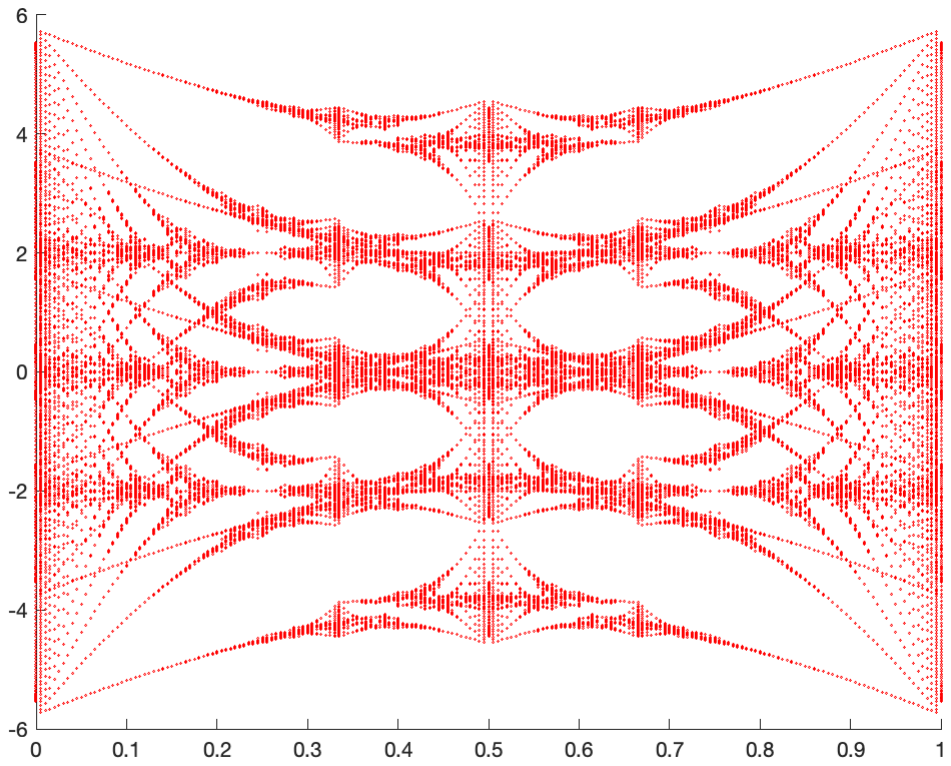} 
\includegraphics[width=7.5cm,height=8cm, trim=0.5cm 6cm 0.5cm 6cm]{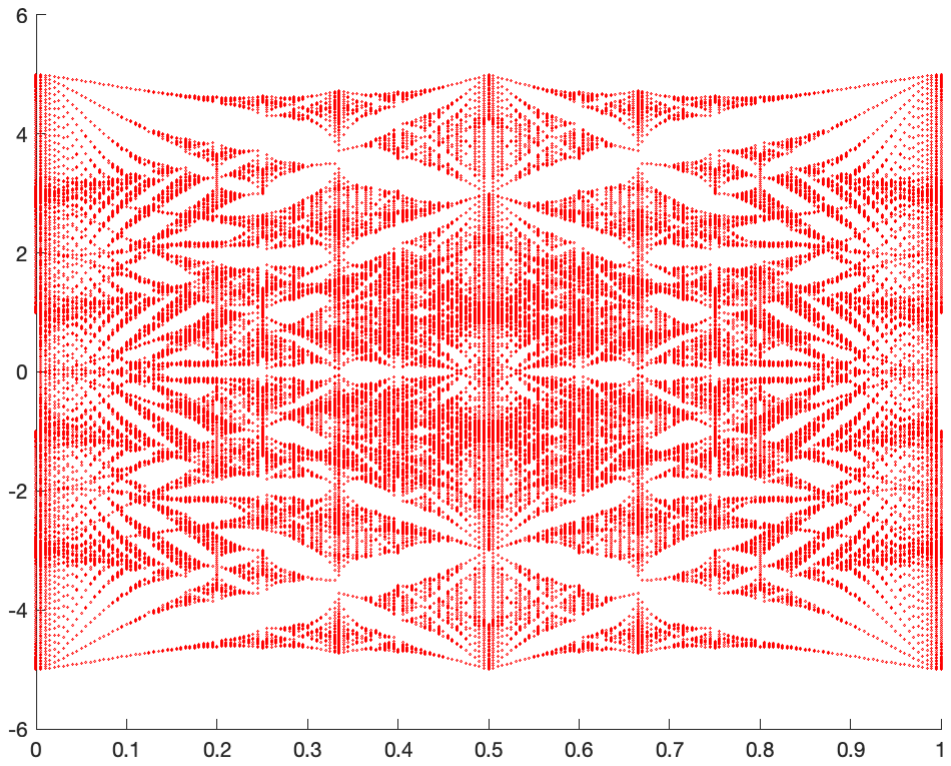}
\caption{\label{fig:AC} Hofstadter butterflies (1/L-Spectrum plots). The left figure shows the spectrum for the anti-chiral ($w_0=1$) and the right figure for the chiral potential ($w_1=1$), both for the case $\theta=0$ for $1/L \in [0,1].$}
\end{figure}

Using standard results about the almost Mathieu operator, we readily conclude the following proposition partially explaining the occurrence of zero measure and fat Cantor spectra in Figures \ref{fig:AC} and \ref{fig:fat_cantor}. 

\begin{figure}
\includegraphics[width=7.5cm,height=7cm]{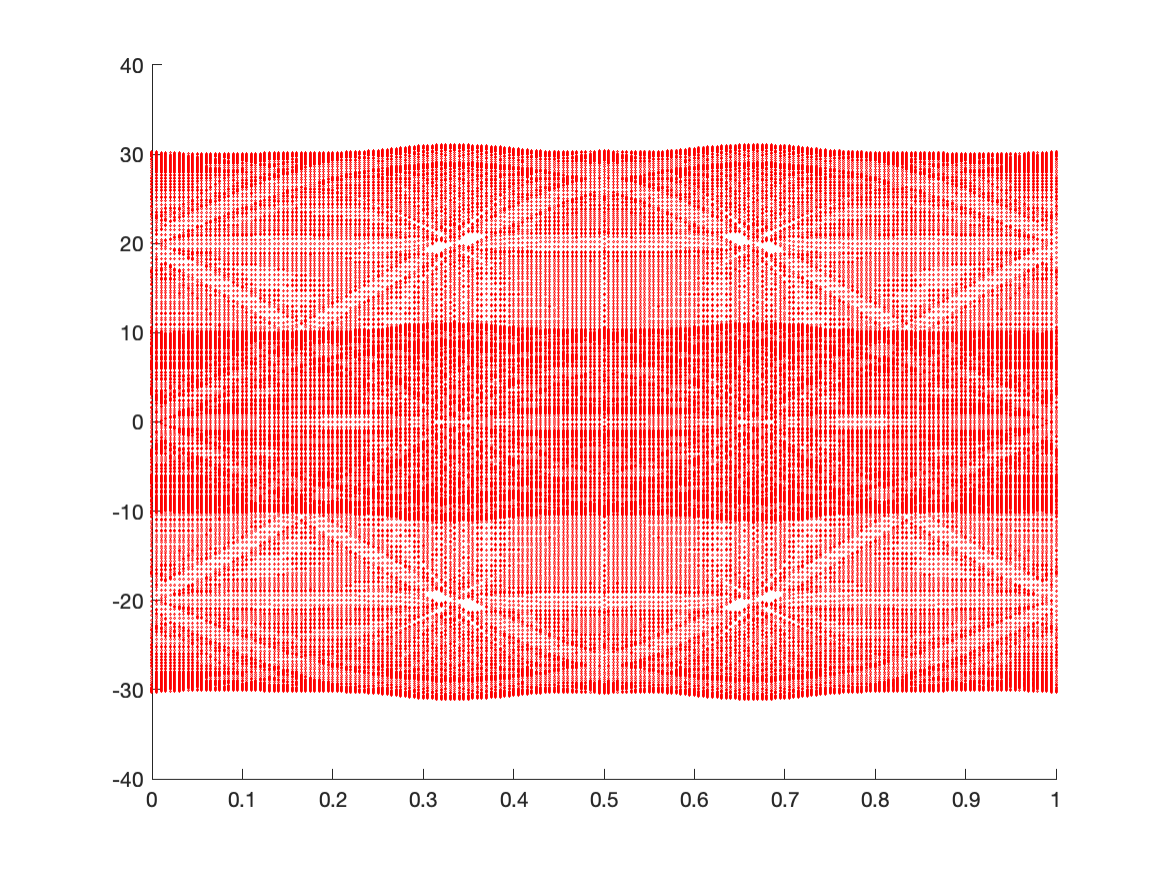} 
\caption{\label{fig:fat_cantor} Hofstadter butterflies (1/L-Spectrum plots). Spectrum of chiral Hamiltonian with subcritical tunneling $w_1= \frac{2}{5}.$ Similarly to the AMO with coupling constant away from the critical coupling, the spectrum starts to become more dense.}
\end{figure}

\begin{prop}[AMO]
\label{prop:AMO}
Let $1/L \notin \mathbb Q, \theta \in \{0,1/2,1\}$. Let $w_0 = 3$, then the spectrum of the anti-chiral Hamiltonian is purely singular continuous and a Cantor set of zero measure. Let $w_0>3$ then the spectral measure is absolutely continuous and for $w_0<3$ it is pure point if $1/L$ satisfies a diophantine condition. In either case, the spectrum is a Cantor set of positive measure.
\end{prop}
\begin{proof}
The spectrum, as a set, of a finite direct sum of operators is the finite union of the spectra of all operators on the diagonal. These are in all cases we consider here Cantor sets, i.e. closed, nowhere dense sets, without isolated points. Thus, the finite union of such sets will still be Cantor sets and clearly the spectral type is also preserved under finite direct sums of operators. The same argument applies to the Lebesgue decomposition of the spectral measure.
\end{proof}

\section{Spectral analysis of flat bands in effective models}
\label{sec:spec_ana}

In this section, we compare the spectral analysis of the discrete model with an effective low energy model, valid close to zero energy, that has been proposed in \cite{TM2020} and subsequently been analyzed in \cite{BW22}. As a main difference, this model does not exhibit any quasiperiodic features and is therefore closer in spirit to the continuum model of twisted bilayer graphene. 

The proposed continuum Hamiltonian is given by  \[  \mathscr L = \begin{pmatrix} 0 & D_x -i k_{\perp} & w_0 U(x/L) & w_1  U^-(x/L) \\ D_x+i k_{\perp} & 0 & w_1 U^+(x/L) & w_0U(x/L)  \\ w_0 U(x/L) & w_1 U^+(x/L) & 0 & D_x -i k_{\perp}   \\  w_1 U^-(x/L) & w_0 U(x/L) & D_x +i k_{\perp} & 0 \end{pmatrix}.\] 
Since this Hamiltonian $ \mathscr L$ is periodic, we can apply standard Bloch-Floquet theory to equivalently study the spectrum on $L^2(\mathbb R/L\ZZ)$ of 

\[  \mathscr L(k_x) = \begin{pmatrix} 0 & D_x+k_x -i k_{\perp} & w_0 U(x/L) & w_1  U^-(x/L) \\ D_x+k_x +i k_{\perp} & 0 & w_1 U^+(x/L) & w_0U(x/L)  \\ w_0 U(x/L) & w_1 U^+(x/L) & 0 & D_x+k_x -i k_{\perp}   \\  w_1 U^-(x/L) & w_0 U(x/L) & D_x+k_x +i k_{\perp} & 0 \end{pmatrix}\] such that 
\begin{equation*}
\Spec( \mathscr L) = \bigcup_{k_x \in [0,2\pi/L]} \Spec( \mathscr L(k_x)).
\end{equation*}

The study of the nullspace of the Hamiltonian $ \mathscr L(k_x)-\lambda$ is equivalent to the study of the nullspace of the operator $\widehat{ \mathscr L}_{\lambda}(k_x)=\operatorname{diag}(\sigma_1,\sigma_1)( \mathscr L(k_x)-\lambda)$ given by
\begin{equation}\label{eq:tm20afterconj}
\widehat{ \mathscr L}_{\lambda}(k_x)=\begin{pmatrix} D_x+k_x+ik_{\perp} & w_1 U^+(x/L) & -\lambda& w_0U(x/L) \\w_1 U^-(x/L) & D_x+k_x+ik_{\perp} &w_0U(x/L) & -\lambda \\ -\lambda & w_0U(x/L) & D_x+k_x-ik_{\perp} &w_1 U^-(x/L) \\ w_0 U(x/L) & -\lambda & w_1 U^+(x/L) & D_x+k_x-ik_{\perp} \end{pmatrix} .
\end{equation}

\begin{prop}
\label{prop:no_flat}
The Hamiltonian $ \mathscr L$ does not possess any flat bands in $k_x$ for any fixed $k_{\perp}$, i.e. there is no $\lambda \in \mathbb R$ such that $\lambda \in \Spec( \mathscr L(k_x))$ for all $k_x \in \RR.$
\end{prop}
\begin{proof}
This is an easy consequence of \eqref{eq:tm20afterconj} and Bloch-Floquet theory. In fact, by \eqref{eq:tm20afterconj} we have $\lambda\in\Spec( \mathscr L(k_x))$ if and only if $k_x\in\Spec(\widehat{ \mathscr L}_\lambda(0))$. Since the Hamiltonian $\widehat{ \mathscr L}_\lambda(0))$ has compact resolvent, its spectrum is discrete and therefore it is impossible that $k_x\in\Spec(\widehat{ \mathscr L}_\lambda(0))$ for all $k_x \in \mathbb R$, which proves the claim.
\end{proof}

We shall now restrict us to the case $k_{\perp}=0$ and study the spectrum of the continuous Hamiltonian. In particular, we shall analyze under what conditions $0$ is in the spectrum.
We start with the anti-chiral Hamiltonian for which our spectral analysis is rather complete.
\begin{prop}
The spectrum of the anti-chiral Hamiltonian $ \mathscr L_{\operatorname{ac}}(k_x)= \mathscr L(k_x)$ with $w_1=0$ for $k_{\perp}=0$ satisfies $\bigcup_{k_x\in[0,2\pi/L]} \Spec( \mathscr L_{\operatorname{ac}}(k_x)) =\RR$. In particular, $0\in\Spec( \mathscr L_{\operatorname{ac}})$ for all $w_0\in\mathbb R$.
\end{prop}
\begin{proof}
According to \eqref{eq:tm20afterconj} it suffices to consider the equation
\begin{equation}\label{eq:ODE}
D_x \varphi + A(\lambda,x)\varphi= 0,
\end{equation}
where
\[ A(\lambda,x) = \begin{pmatrix} k_x &0 &  -\lambda &  w_0U(x/L) \\ 0 & k_x & w_0U(x/L)& -\lambda \\  -\lambda & w_0U(x/L)& k_x & 0 \\ w_0U(x/L) & -\lambda &0 & k_x \end{pmatrix}.\]
Let $W(x/L) = i\frac{L}{3}( x/L + \frac{1}{\pi} \sin( 2\pi x/L ) )$, so that $D_xW(x/L)=U(x/L)$. With
\[ B(\lambda,x) = \begin{pmatrix} ik_x  x&0 &  -i\lambda x&  w_0W(x/L) \\ 0 & ik_x x & w_0W(x/L)& -i\lambda x \\  - i \lambda x & w_0W(x/L)& i k_x x& 0 \\ w_0W(x/L) & - i\lambda x&0 & ik_x x \end{pmatrix}.\]
we thus have $D_xB(\lambda,x)=A(\lambda, x)$. Using the unitary matrix 
$$
\mathscr U = \frac{1}{2} \begin{pmatrix} -1& 1 & -1 & 1 \\ -1 & -1 & 1 & 1 \\ 1 & -1 & -1 & 1 \\ 1 & 1 &1 & 1 \end{pmatrix}, 
$$
we see that 
\[\begin{split}\mathscr U B(\lambda,x) \mathscr U^*= \operatorname{diag}&(i\lambda x-w_0{W(x/L)}, i\lambda x+w_0{W(x/L)},\\ & -i\lambda x-w_0{W(x/L)}, -i\lambda x+w_0{W(x/L))} +i k_x x \end{split}\]
is diagonal. Since $\mathscr U e^{-B} \mathscr U^\ast =e^{-\mathscr U B\mathscr U^\ast}$ it then follows that
$$
D_xe^{-B}=\mathscr U^\ast D_x e^{-\mathscr U B\mathscr U^\ast}\mathscr U=\mathscr U^\ast D_x(-\mathscr UB\mathscr U^\ast)e^{-\mathscr U B\mathscr U^\ast}\mathscr U=-Ae^{-B}.
$$
Hence, the solutions to \eqref{eq:ODE} are of the form $\varphi(x) =  e^{-B(\lambda,x)}\varphi_0.$
For $\lambda$ to be an eigenvalue, such a solution is required to be $L$-periodic. Inspecting the expression for $\mathscr U B(\lambda,x) \mathscr U^*$ we see that it is necessary that, for $w_0$ fixed, we can find $k_x \in [0,2\pi/L]$ such that for any of the combinations $\lambda \pm \frac13w_0\pm k_x  \in \frac{2\pi}{L} \ZZ$. Thus $\Spec( \mathscr L_{\operatorname{ac}})=\RR.$
 \end{proof}

For the chiral Hamiltonian $ \mathscr L_c =  \mathscr L$ with $w_0=0$, we do not have an explicit description of the full spectrum, however we can still locate the zero energy spectrum.

\begin{prop}
For all $w_1 \in \mathbb R$ it follows that $0 \in \Spec( \mathscr L_c)$ but there is $\varepsilon>0$ such that for all $w_1 \in (-\varepsilon,\varepsilon)\backslash \{0\}$ we have $0 \notin \Spec(\mathscr L_c(k_x=0)).$
\end{prop}
\begin{proof}
To prove the first part of the statement, note that it suffices to show that there is $\mu \in \mathbb C$ such that for some function $\psi \in H^1(\RR/\ZZ; \CC^2)$
\[ D_x \psi  +w_1 \begin{pmatrix} 0 & U^+(x) \\ U^-(x) & 0 \end{pmatrix} \psi = \mu \psi.\]
Let then $x \mapsto X(x,w_1)$ be the fundamental solution with $\mathcal U(x):=  \begin{pmatrix} 0 &  U^+(x) \\ U^-(x) & 0 \end{pmatrix} $, satisfying
$$ D_x X(x) + w_1\mathcal U(x) X(x) = 0, \quad X(0) = \operatorname{id}.$$
The matrix $M(w_1):=X(1,w_1)$, with $\operatorname{det}(M)=1$, is then called the monodromy matrix and let $\rho \in \Spec(M(w_1))$ where $\rho \neq 0.$
Thus, there is $v \neq 0$ such that $\phi(x):=X(x)v$ satisfies the periodicity condition $\phi(1) = M(w_1)v =\rho v= \rho X(0,w_1) v = \rho \phi(0).$ We then define $\mu \in \CC$, such that
$\psi(x):=e^{i\mu x} \phi(x)$ is the desired solution.

If $\lambda=0$ was protected in the $k_x=0$ sector, then this would imply that there is always a solution
\[ D_x \psi  +w_1\mathcal U(x) \psi =0.\]
In this case, our claim amounts to showing that $1 \notin \Spec(M(w_1))$ for $w_1$ small but non-zero.
Since $D_x (\operatorname{det}(X(x,w_1))) = w_1 \tr(\mathcal U(x)) = 0$ it follows that $\operatorname{det}(X(x,w_1)) = 1.$
On the other hand, since $X(\bullet,x)$ is analytic, such that $X(\bullet,x) = \sum_{i \ge 0} w_1^i X_{i}(x),$ we find the recurrence equation 
\[ D_x X_{i+1}(x) = \mathcal U(x) X_{i}(x), \quad X_{i+1}(0) = 0,\]
and $X_{0} = \operatorname{id}_{\mathbb C^2}.$
Hence, all $X_{i}$ with $i$ odd are trace-free. We compute $\tr(X_{2}(1)) =1.$ 
This implies that $\tr(M(w_1)) = \tr(X_0) + w_1^2 \tr(X_2) + \mathcal{O}(w_1^4) \neq 2$ for $w_1 $ small but non-zero. Hence, $1 \notin \Spec(M(w_1)).$
\end{proof}

\section{A two-dimensional example}
\label{sec:2D}

\begin{figure}
\begin{center}
\includegraphics[height=5cm,width=5cm]{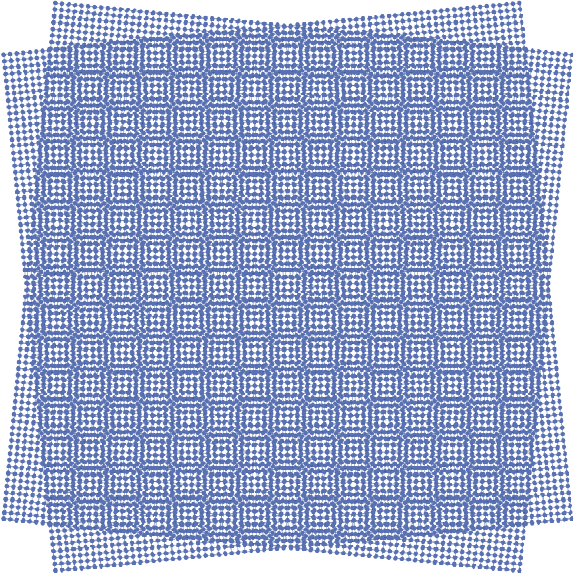}
\caption{\label{fig:square} Twisted square lattices exhibit yet another macroscopic (moir\'e) square lattice.}
\end{center}
\end{figure}

We now consider the case of $2D$ twisted lattice structures. For simplicity, we shall consider two square lattices with moir\'e lengths $L_1,L_2 >0$, as discussed for example in \cite{KV19}. Figure \ref{fig:square} shows an example of such a moir\'e pattern. The kinetic energy is described by a discrete Laplacian on each of the lattices in terms of
\[ (-\Delta_{\ZZ^2} u)_n = (u_{n+e_1}+u_{n-e_1} +u_{n+e_2}+u_{n-e_2})\text{ with } n \in \ZZ^2 \]
such that $(D_{\operatorname{kin}} \psi)_n =- (\operatorname{diag}(\Delta_{\ZZ^2},\Delta_{\ZZ^2}) \psi)_n$ is the discrete Laplacian of the individual lattices without any additional interaction. The interaction is then modeled by a tunneling potential $$V_w(n) = w \begin{pmatrix} 0 & U(\tfrac{n_1}{L_1},\tfrac{n_2}{L_2})\\ U(\tfrac{n_1}{L_1},\tfrac{n_2}{L_2})& 0 \end{pmatrix}$$
with coupling strength $w >0,$ where we assume that $U$ is a real-valued smooth $1$-periodic function in both components. This defines a Hamiltonian $H : \ell^2(\ZZ^2;\CC^4) \to \ell^2(\ZZ^2; \CC^4)$
\begin{equation*}
(H \psi)_n = (D_{\operatorname{kin}}\psi)_n+V_w \psi_n.
\end{equation*}
We then introduce
\[ P = \left(P_{\tfrac{n_1}{L_{1}},\tfrac{n_2}{L_{2}}}\right)_{n \in \ZZ^2} \text{ with }P_X =\frac{1}{\sqrt{2}} \begin{pmatrix} -\sgn(U(X)) & 1 \\ 1 & \sgn(U(X))  \end{pmatrix},\] then conjugating by $P$ yields, after swapping entries according to the sign of $U$, an equivalent block-diagonal Hamiltonian $\tilde H:\ell^2(\ZZ^2;\CC^4) \to \ell^2(\ZZ^2;\CC^4)$
\begin{equation}
\label{eq:2dH}
\tilde H\psi_{n} = \operatorname{diag}\left(-\Delta_{\ZZ^2} +  U(\tfrac{n_1}{L_1},\tfrac{n_2}{L_2}) , -\Delta_{\ZZ^2} -  U(\tfrac{n_1}{L_1},\tfrac{n_2}{L_2}) \right)\psi_n.
\end{equation}

This leads us to the following result which shows that for a set of moir\'e length scales of large measure, the model actually exhibits Anderson localization. 

\begin{theo}\cite[Theorem $6.2$]{BGS02}
Let $U: \mathbb T^2 \to \mathbb R$ be real analytic such that the marginals $$\theta_1 \mapsto U(\theta_1,\theta_2) \text{ and } \theta_2 \mapsto U(\theta_1,\theta_2)$$
are non-degenerate. Moreover, let $\varepsilon>0,$ then for any $w \ge w_0(\varepsilon)$ there exists a set of moir\'e length scales $(1/L_1,1/L_2) \in \mathbb T^2$ of measure $1-\varepsilon$ such that the Hamiltonian exhibits full Anderson localization, i.e.~the spectrum consists only of exponentially decaying eigenfunctions. 
\end{theo}
Related questions on the existence of absolutely continuous spectrum for small coupling $w$ and Cantor spectrum are widely open.
The situation simplifies, once one imposes a separability condition, $U(x)= U_1(x_1) U_2(x_2)$ with $U_1,U_2$ real-analytic and non-degenerate. In this case, the operator \eqref{eq:2dH} decomposes into the direct sum of two Hamiltonians 
\[ H_1 = \operatorname{diag}(-\Delta_{\ZZ} + U_1, -\Delta_{\ZZ}-U_1 ) \text{ and } H_2 = \operatorname{diag}(-\Delta_{\ZZ} + U_2, -\Delta_{\ZZ}-U_2 ) .\]

\section*{Acknowledgments}
Lingrui Ge was partially supported by NSFC grant (12371185) and the Fundemental Research Funds for the Central Universities (the start-up fund), Peking University. The research of Jens Wittsten was partially supported by The Swedish Research Council grants 2019-04878 and 2023-04872.

Rui Han and Wilhelm Schlag \cite{HS} have noticed and carefully analyzed a technical error in Section \ref{sec:AL} of our first version. In connection to this we want to kindly point out that the model we study in Subsection \ref{sec:AVAL} relies crucially on having just a cosine tunneling potential. To increase clarity, this additional assumption is now stressed in Assumption \ref{ass:assumption}, while the symmetry error claimed in \cite{HS} refers to the potential $U(x)=\frac{1+2\cos(2\pi x)}{3}$, which is not the one we consider in Subsection \ref{sec:AVAL} where we make our symmetry claims. (In \cite[\S 7.1]{HS} the potential $U_{\operatorname{c}}^+$ is modified compared to \cite{TM2020} and \eqref{eq:scalarpotentials} which also leads to some simplifications and allows application of \cite{K17} even in the case $w_0=w_1$, which is not possible for the original \cite{TM2020} model as explained in Remark \ref{rem:constantEV}.) Nevertheless, we would like to thank \cite{HS} for their analysis and attention to our first version which has helped us to find some misprints and clarify our writing. We also wish to thank two anonymous referees for their comments and suggestions.

%
%

\end{document}